\documentclass[12pt]{amsart}
\usepackage{amsmath,amssymb, amscd}

\usepackage{graphicx}

\usepackage{amscd}
\usepackage[all]{xy}
\usepackage{xcolor}

\makeatletter
\@addtoreset{equation}{section}
\makeatother
\usepackage{enumerate}
\usepackage{mathrsfs}
\providecommand{\texorpdfstring}[2]{#1}

\usepackage{color}

\def\CC{{\mathbb C}}

\def\ZZ{{\mathbb Z}}

\def\PP{{\mathbb P}}

\def\bH{{\mathbf H}}

\def\be{{\mathbf e}}

\def\ee{{\mathrm e}}

\def\ii{{\mathfrak{i}}}

\def\cQ{{\mathcal Q}}

\def\fD{{\mathfrak D}}

\def\fS{{\mathfrak S}}

\def\fa{{\mathfrak a}}

\def\tX{{\widetilde X}}

\def\tX{{\widetilde X}}

\def\tv{{\widetilde v}}

\def\tbeta{{\widetilde \beta}}

\def\tomega{{\widetilde \omega}}

\def\hw{{\widehat w}}

\def\hz{{\widehat z}}

\def\hB{{\widehat B}}

\def\hJ{{\widehat J}}

\def\hP{{\widehat P}}

\def\hX{{\widehat X}}

\def\hgamma{{\widehat \gamma}}

\def\hkappa{{\widehat \kappa}}

\def\hnu{{\widehat \nu}}

\def\hvarpi{{\widehat \varpi}}

\def\cn{{\mathrm {cn}}}
\def\sn{{\mathrm {sn}}}
\def\dn{{\mathrm {dn}}}

\def\al{{\mathrm {al}}}

\def\nuI#1{{\nu^{\mathrm{I}}_{#1}}}

\def\nuII#1{{\nu^{\mathrm{II}}_{#1}}}

\def\trp{{\, {}^t\negthinspace}}

\newtheorem{definition}{Definition}[section]

\newtheorem{theorem}[definition]{Theorem}
\newtheorem{proposition}[definition]{Proposition}
\newtheorem{corollary}[definition]{Corollary}
\newtheorem{remark}[definition]{Remark}
\newtheorem{lemma}[definition]{Lemma}

\begin{document}

\title[The $\al_{ab}$ functions and a degenerate Krichever-Novikov equation]{The $\al_{ab}$ functions of hyperelliptic curves: a generalization of Jacobi's elliptic functions and explicit solutions of a degenerate Krichever-Novikov equation}

\author{Shigeki Matsutani}

\date{\today}



\begin{abstract}
The Jacobi elliptic functions $\sn$, $\cn$, $\dn$ and their ratios are expressed by Weierstrass's functions $\al_a=\sqrt{x-b_a}$ of an elliptic curve.
For a hyperelliptic curve of genus $g$, Weierstrass's $\al_a$ functions generalize $\sn$, $\cn$ and $\dn$, and the modified Korteweg-de Vries and sine-Gordon equations stand behind them.
In this paper we study the $\al_{ab}$ functions, which are defined by the sigma function shifted by the sum of two half-periods and which, together with the ratios of the $\al_a$ functions, generalize the remaining Jacobi elliptic functions.
We determine the exact lattice of periods of $\al_{12}$ and show that it is single-valued on the unramified double covering of the Jacobian corresponding to the double covering $\hw^2=(x-b_1)(x-b_2)$ of the curve.
We show that $\al_{12}$ is the Wronskian of the band-edge eigenfunctions $\al_1$ and $\al_2$, so that the half-period translations realize Darboux transformations.
As the main result, $\al_{12}$, $\al_2/\al_1$ and $1/\al_{12}$ satisfy one and the same degenerate Krichever-Novikov equation, whose polynomial has the double roots $0$ and $\infty$, and the half-periods act on these solutions as the M\"obius symmetries of the equation.
Thus the Krichever-Novikov equation stands behind the $\al_{ab}$ functions.
In genus one these functions reduce to $\sn$, $\mathrm{ns}$ and $\mathrm{dc}$.
\end{abstract}

\maketitle

\section{Introduction}\label{sec1}

The Jacobi elliptic functions $\sn$, $\cn$ and $\dn$, together with the nine further functions obtained as their ratios and reciprocals (the twelve Jacobi elliptic functions), are among the most frequently used special functions in the description of nonlinear phenomena.
Their origin is Abel's elliptic functions $\varphi$, $f$ and $F$ \cite{Abel}.
In Weierstrass's theory they are expressed by the functions $\al_a(u)=\sqrt{x-b_a}$ $(a=1,2,3)$ of the elliptic curve $y^2=(x-b_1)(x-b_2)(x-b_3)$, $\displaystyle{u=\int^{(x,y)}_\infty \frac{dx}{2y}}$, which are written by the elliptic sigma function as \cite{Wei85}
$$
\al_a(u)=\frac{\sigma_a(u)}{\sigma(u)}, \quad
\sigma_a(u)=\frac{\ee^{-\eta_a u}\sigma(u+\omega_a)}{\sigma(\omega_a)},
$$
where $\omega_a$ and $\eta_a$ are the half-periods of the first and the second kind associated with the branch point $B_a:=(b_a,0)$.
In fact, with $z=\sqrt{b_1-b_3}\,u$ and the modulus $k^2=(b_2-b_3)/(b_1-b_3)$,
\begin{equation}
\sn(z) = \frac{\sqrt{b_1-b_3}}{\al_3(u)}, \quad
\cn(z) = \frac{\al_1(u)}{\al_3(u)}, \quad
\dn(z) = \frac{\al_2(u)}{\al_3(u)},
\label{eq:Jsn}
\end{equation}
so that all twelve Jacobi elliptic functions are ratios of $\al_1$, $\al_2$, $\al_3$ and constants.

In 1854 Weierstrass introduced the hyperelliptic $\al_a$ functions to solve the Jacobi inversion problem for the hyperelliptic curve $X$: $y^2=\prod_{i=1}^{2g+1}(x-b_i)$ of genus $g$ \cite{Wei54}.
For a point $((x_1,y_1),\ldots,(x_g,y_g))$ of the symmetric product $S^gX$, $\al_a$ is, up to a constant factor, $\sqrt{(x_1-b_a)\cdots(x_g-b_a)}$, and in terms of the hyperelliptic sigma function of Klein \cite{Klein86}
\begin{equation}
\al_a(u)=\gamma_a''\frac{\ee^{- \trp u \eta_{B_a}}
       \sigma{}(u +  \omega_{B_a})}
{\sigma{}(u) \sigma_{\natural_1}(\omega_{B_a})},
\label{eq:In_al1}
\end{equation}
where $\omega_{B_a}$ and $\eta_{B_a}$ are the half-periods of the first and the second kind associated with $B_a=(b_a,0)$; the precise definitions are given in Sections 2 and 3.
Thus the $\al_a$ functions are the hyperelliptic generalization of $\sn$, $\cn$ and $\dn$ \cite{Mat02c, Mat02a, M25}.
They also appear, without the name, in Mumford's treatment of the sine-Gordon equation and of the Neumann and Mumford dynamical systems \cite[3.69, 3.81, 3.241]{Mum84}, which have recently been studied in connection with the hyperelliptic sigma function \cite{Bu, Ba1, Ba2}.

The elliptic functions $\sn$, $\cn$ and $\dn$ are the classical solutions of the static sine-Gordon and modified Korteweg-de Vries (MKdV) equations.
The hyperelliptic $\al_a$ functions keep this property in every genus.
Namely, $\al_a$ is the eigenfunction at the band edge $b_a$ of the Schr\"odinger operator $L=\partial_{u_g}^2-2\wp_{gg}(u)$ with the finite-gap potential $2\wp_{gg}$ \cite{D, DT, GH, BBEIM}, its logarithmic derivative $k=-2\ii\,\partial_{u_g}\log\al_a$ satisfies the MKdV equation over $\CC$,
\begin{equation}
\partial_{t}k+\frac{3}{2}k^2 \partial_s k+\partial_{s}^3 k=0,
\label{4eq:MKdV_k}
\end{equation}
with $s=u_g$ and $t$ a linear combination of $u_{g-1}$ and $u_g$ \cite{Mat02c, M25} (Remark \ref{rmk:al_aNLS_CMKdV}), the relation between $k$ and $2\wp_{gg}$ being the Miura transformation, and $\al_a^2$ satisfies the sine-Gordon equation \cite[3.241]{Mum84}, \cite{Mat02a, M25} (see (\ref{4eq:SG1-4a})).
In other words, the MKdV and sine-Gordon equations stand behind the $\al_a$ functions, as their static versions stand behind $\sn$, $\cn$ and $\dn$.

\bigskip

The $\al_{ab}$ functions also go back to Weierstrass, who introduced $\mathrm{al}(u)_{ab}$ together with its behaviour under the half-period translations \cite[(7), (17)]{Wei54}.
Bolza and Baker studied the corresponding shifts $\sigma(u+\omega_{B_{a_1}}+\cdots+\omega_{B_{a_r}})$ of the sigma function by sums of half-periods \cite{Bol95, Baker98}, and following them we wrote in \cite{M25}
$$
\al_{12}(u):=\gamma_{12}'' \frac{
\ee^{-\trp u  (\eta_{B_1}+\eta_{B_2})}
\sigma( u + \omega_{B_1}+ \omega_{B_2})}{\sigma(u)\sigma_{\natural_2}(\omega_{B_1}+\omega_{B_2})}.
$$
In genus one, $\omega_{B_1}+\omega_{B_2}$ is congruent to $\omega_{B_3}$ modulo the lattice, and $\al_{12}=\pm\al_3$.
Hence the $\al_a$ functions, the $\al_{ab}$ functions and their ratios form the hyperelliptic counterpart of the whole set of the twelve Jacobi elliptic functions; for instance, $1/\al_{12}$, $\al_{12}$ and $\al_2/\al_1$ reduce to $\sn$, $\mathrm{ns}$ and $\mathrm{dc}$ up to constant factors (Appendix \ref{sec:g1}).
The natural question is then which integrable equation stands behind the $\al_{ab}$ functions, as the MKdV and sine-Gordon equations stand behind the $\al_a$ functions.
The purpose of this paper is to answer this question: it is the Krichever-Novikov equation.

The Krichever-Novikov (KN) equation \cite{KN80}
\begin{equation}
q_t=q_{sss}-\frac32\,\frac{q_{ss}^2-h(q)}{q_s},\qquad
\frac{q_t}{q_s}=\{q;s\}+\frac32\,\frac{h(q)}{q_s^2},
\label{eq:KN_intro}
\end{equation}
where $h$ is a polynomial of degree at most four and $\{q;s\}=q_{sss}/q_s-\frac32(q_{ss}/q_s)^2$ is the Schwarzian derivative, was found in the study of commuting differential operators of rank two and of the finite-gap solutions of the Kadomtsev-Petviashvili equation.
In the classification of integrable third order evolution equations it occupies a special place \cite{SS82, MSY87, HSS95}:
when the roots of $h$ (regarded as a binary quartic) are simple, the KN equation is not related to any other equation of KdV type by a differential substitution \cite{SSY83}, although it admits a B\"acklund transformation of the form $q_s\widetilde{q}_s=p(q,\widetilde{q})$ \cite{Adler98}, and its algebro-geometric solutions were constructed by D. P. Novikov by means of a Prym variety of a double covering of an elliptic curve \cite{Novikov99}.
When $h$ has multiple roots (the degenerate case), the KN equation is mapped to the KdV equation by an explicit differential substitution \cite{SSY83} (see also \cite{EE09}).
For $h=0$ it is the Schwarzian KdV equation, which is satisfied by the ratio of two solutions of the Lax pair of the KdV equation with the same spectral parameter \cite{Weiss83}; theta-function solutions of the Schwarzian KdV equation were obtained from higher-order degenerations of Fay's identity, where a branch point of a hyperelliptic curve is used \cite{KP23}.

\bigskip

In this paper we investigate the geometric, algebraic and analytic properties of the $\al_{ab}$ functions in detail.
Our results are the following.
\begin{enumerate}
\item[(i)] $\al_{12}$ changes sign under the lattice $\Gamma_X$ by an explicit character (Lemma \ref{lm:multiplier}); it is single-valued on the unramified double covering of $J_X$ corresponding to the double covering $\hX_{12}$: $\hw^2=(x-b_1)(x-b_2)$ of $X$, and it is a meromorphic function on $S^g\hX_{12}$ (Proposition \ref{pr:al12=ww}).
\item[(ii)] $\al_{12}$ is, up to a constant, the Wronskian of the band-edge eigenfunctions $\al_1$ and $\al_2$ of $L$, and the half-period translations realize the Darboux transformations at the band edges (Remark \ref{rmk:Crum}); this is an explicit form of the classical fact that a Darboux transformation shifts the divisor on the Jacobian \cite{BC28, EK82, GH00}.
\item[(iii)] The three functions $q=\al_{12}$, $q=\al_2/\al_1$ and $q=1/\al_{12}$ satisfy one and the same degenerate KN equation (Theorem \ref{th:KN})
$$
4\partial_{u_{g-1}}q=q'''-\frac{3}{2}\,\frac{(q'')^2-(b_1-b_2)^2q^2}{q'}+\bigl(2\lambda_{2g}+3(b_1+b_2)\bigr)q',\qquad q':=\partial_{u_g}q ,
$$
that is, (\ref{eq:KN_intro}) with $h(q)=(b_1-b_2)^2q^2$, whose roots $0$ and $\infty$ are both double.
The half-period translations act on the solutions as the M\"obius transformations $q\mapsto cq$ and $q\mapsto 1/q$ preserving the roots of $h$ (Remark \ref{rmk:KN}).
\item[(iv)] The functions $\al_2/\al_1$ and $\al_{12}$ are related by Adler's B\"acklund transformation with $p(q,\widetilde{q})=(b_2-b_1)q\widetilde{q}$, and the substitution of \cite{SSY83} maps $\al_2/\al_1$ to $-12\wp_{gg}(u+\omega_{B_1})$ up to an additive constant, i.e., to a translate of the finite-gap KdV potential (Remark \ref{rmk:KN}).
\item[(v)] In genus one, the equation of (iii) reduces to an ordinary differential equation satisfied by $\sn$, $\mathrm{ns}$ and $\mathrm{dc}$ (Appendix \ref{sec:g1}).
\end{enumerate}
The degenerate KN equation itself is, of course, not new \cite{KN80, SS82, SSY83}.
What is new is that the $\al_{ab}$ functions provide its explicit algebro-geometric solutions for every genus $g$, expressed both by the sigma function and by the coordinates of $S^g\hX_{12}$, together with the description of their periods, of the double covering on which they live, and of the action of the half-periods.
Compared with \cite{Weiss83, KP23}, where the ratio of two solutions with the same spectral parameter leads to $h=0$, our functions are built from the eigenfunctions at two different band edges $b_1\ne b_2$, and the difference $b_1-b_2$ of the eigenvalues appears in $h(q)=(b_1-b_2)^2q^2$.
Thus we obtain the following picture: the logarithmic derivative of $\al_a$ is related to the KdV potential $2\wp_{gg}$ by the Miura transformation and satisfies the MKdV equation, while $\al_{ab}$ and $\al_b/\al_a$ are related to the translates of the same potential by the substitution of \cite{SSY83} and satisfy the degenerate KN equation.
In this sense the KN equation stands behind the $\al_{ab}$ functions, as the MKdV and sine-Gordon equations stand behind the $\al_a$ functions, and as their static versions stand behind the Jacobi elliptic functions.

\bigskip

The structures used in this paper are classical; we refer to \cite[Chapter 1]{M25} for their history and mention only a few points.
Gauss wrote $\sn=P/Q$ with entire functions $P$, $Q$ obeying $PP''-P'^2=-Q^2$ and $QQ''-Q'^2=-\mu P^2$, which is Hirota's bilinear form \cite{Hirota71}, and derived $(\log\sn)''=-1/\sn^2+\mu\,\sn^2$, the genus one case of (\ref{4eq:SG1-4a}) \cite{Gauss}; as Fricke remarks there, $P$ and $Q$ are Weierstrass's $\mathrm{Al}(u)_1$ and $\mathrm{Al}(u)$ \cite{W40}.
In \cite[p.~296]{Wei54} Weierstrass proposed to write $\mathrm{al}(u)_a$ and $\mathrm{al}(u)_{ab}$ as quotients $p_a/p$ and $p_{ab}/p$ of entire functions satisfying differential equations of the second order, the programme which led to the sigma function, whose symbol stems from Gauss's $S$ in the German typeface Fraktur $\fS$.
Baker's determinant formulas for the shifts of $\sigma$ by several half-periods \cite[Section IX]{Baker98} give Crum's Wronskian of the $\al_a$ functions for three half-periods (Proposition \ref{pr:Baker_Crum}), and the same is strongly expected in general (Remark \ref{rmk:BakerCrum}).

This paper belongs to the study of hyperelliptic functions by means of the sigma function in the line of Klein and Baker \cite{Klein86, Baker98, Baker03} developed by Buchstaber, Enolskii and Leykin \cite{BEL, BEL20}, in which the Abelian functions are expressed by the coordinates of $S^gX$ and their differential identities are derived from the Dubrovin-type equations \cite{D, DT, DMN76, GH}, in close connection with the finite-gap theory of Novikov, Dubrovin, Matveev and Krichever \cite{DMN76, Kri77}.
In this framework Buchstaber and his collaborators constructed polynomial dynamical systems generating hyperelliptic functions and the KdV hierarchy and its deformations \cite{Bu16, AB19}, described the Mumford dynamical system \cite{Bu, Ba1, Ba2}, characterized the sigma functions of families of curves by heat equations \cite{BL04}, and related the hyperelliptic sigma function to the Lie algebras of Schr\"odinger operators and to the Burchnall-Chaundy and Adler-Moser polynomials \cite{BB20, BB21}.
The present paper adds to this line the shifts of the sigma function by half-periods: $\sigma(u+\omega_{B_a})$ and $\sigma(u+\omega_{B_1}+\omega_{B_2})$ give $\al_a$ and $\al_{12}$, the shifts are Darboux transformations in the sense of Burchnall and Chaundy \cite{BC28}, and $\al_{12}$ gives solutions of the equation found by Krichever and Novikov in the theory of commuting operators \cite{KN80}.

Section 2 fixes the notation, Section 3 introduces the $\al_a$ and $\al_{ab}$ functions (Definition \ref{def:als}) and their basic properties, Section 4 reviews the $\al_a$ functions, Section 5 contains the main results (Theorems \ref{th:5.6}, \ref{th:5.7} and \ref{th:KN}), Section 6 is the discussion, and Appendix \ref{sec:g1} treats the genus one case.

\section{Hyperelliptic curve \texorpdfstring{$X$}{X} and sigma function}

In this section we fix the notation for the hyperelliptic curves and the sigma functions; we only state what is needed below and refer to \cite{BEL, BEL20, M25} for the proofs.

\subsection{Geometrical setting of hyperelliptic curve}
Let $X$ be a hyperelliptic curve defined by
\begin{equation}
\begin{split}
X~:~ y^2 &= f(x):= x^{2g+1} + \lambda_{2g} x^{2g} + \cdots +\lambda_0 \\
&= (x-b_1)(x-b_2)(x-b_3) \cdots (x-b_{2g+1})
\label{eq:hyp01}
\end{split}
\end{equation}
together with a smooth point  $\infty$  at infinity. 
Here $\lambda$'s are complex numbers for disjoint points $b_i \in \CC$.
The hyperelliptic involution $\zeta_X : X \to X$ is given as $\zeta_X(x,y)=(x, -y)$.
The affine ring related to $X$ is denoted as $R_X:=\CC[x, y]/(y^2 - f(x))$.
We fix the basis of the holomorphic one-form,
$$ 
\nuI{}:= \left[\begin{matrix} \nuI{1} \\ \vdots \\ \nuI{g}
\end{matrix}\right], \quad
    \nuI{i} =\frac{x^{i-1} d x}{2 y}, \qquad (i=1,\ldots,g),
$$
and the homology basis for the curve  $X$,
$$
\mathrm{H}_1(X, \mathbb Z)
  =\bigoplus_{j=1}^g\mathbb Z\alpha_{j}
   \oplus\bigoplus_{j=1}^g\mathbb Z\tbeta_{j},
$$
as illustrated in Figure \ref{2fig:HypEabII}.
The period matrices are defined by
$$
2\omega_i'= \oint_{\alpha_i}\nuI{}, \quad 2\tomega_i''= \oint_{\tbeta_i}\nuI{},
 \quad (i=1,\ldots,g).
$$
In terms of the branchpoint integrals to the finite ramification point $B_i:=(b_i,0)$ from $\infty$, $\displaystyle{\omega_{B_i}:=\int_{\infty}^{B_i} \nuI{}}$, the elements in the period matrices
for the contours in Figure \ref{2fig:HypEabII} are expressed in the following lemma:
\begin{lemma}\label{2lm:PerodM_hyp1}
$$
2\omega_i'= 2( \omega_{B_{2i-1}}-\omega_{B_{2i}} ), \quad
2\tomega_i''=2 (\omega_{B_{2i}}-\omega_{B_{2g+1}}).
$$
\end{lemma}

Following the convention in \cite[3.81-85]{Mum84}, we define 
$$
2\omega_i''=\oint_{\beta_i} \nuI{}
=2 (\tomega_i''-\sum_{j=i+1}^{g}\omega'_{j}).
$$

The branch-point integrals are expressed by the half-periods as follows \cite{M25}:

\begin{lemma}\label{2lm:PerodM_hyp2}
\begin{equation}
\begin{split}
\omega_{B_{2i}}&=\omega_{B_{2g+1}}+\sum_{j=i+1}^g 
\omega_j'+\frac{1}{2}\sum_{j=i}^g\oint_{\gamma_j} \nuI{}
=-\sum_{j=1}^i\omega_j'+\omega_i'', \quad (i=1,\ldots, g),\\
\omega_{B_{2i-1}}&=-\sum_{j=1}^{i-1}\omega_j'+\omega_i'', \quad (i=1,\ldots, g),\\
\omega_{B_{2g+1}}&=-\sum_{i=1}^g\omega_i'. \\
\end{split}
\label{2eq:Hyp_omega_Bis}
\end{equation}
\end{lemma}

\begin{corollary}\label{cor:omega}
$\omega_{B_{1}}=\omega_1''$ and
$\omega_{B_{2}}=-\omega_1'+\omega_1''$.
\end{corollary}

\begin{figure}[ht]
\begin{center}
\includegraphics[pagebox=mediabox,width=0.7\textwidth]{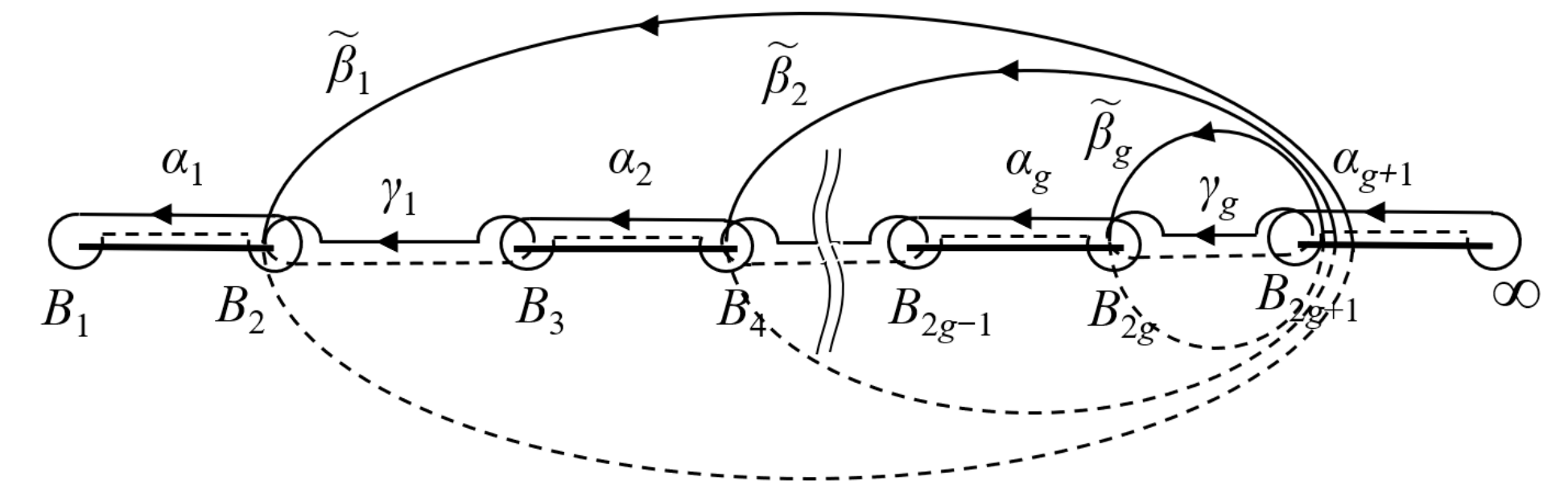}
\end{center}
\caption{
The basis of $\bH_1(X,\ZZ)$}
\label{2fig:HypEabII}
\end{figure}

The period matrices $2\omega'$  and  $2\omega''$ form the lattice $\Gamma_X$ in $\CC^g$ as a $\ZZ$-module.
The Jacobi variety of  $X$  is denoted by $J_X$, $\kappa_J: \CC^g \to J_X=\CC^g/\Gamma_X$.
The Abelian covering of $X$ generated by the path space of $X$ with the basepoint $\infty\in X$ is denoted by $\tX$, $\kappa_X: \tX \to X$, $\kappa_X(\gamma_{P, \infty})=P$, where $\gamma_{P, \infty}$ is a path from $\infty$ to $P\in X$.
For a non-negative integer $k$, we define the Abelian integral $\tv$ from $k$-th symmetric product $S^k \tX$ of $\tX$ to $\CC^g$ and the Abel-Jacobi map $v$ from $S^k X$ to $J_X$ by
\begin{equation*}
\tv:  S^k \tX \to \CC^g, \quad
  \tv(\gamma_1, \ldots, \gamma_k)= \sum_{i=1}^k
       \int_{\gamma_i} 
          \left[\begin{array}{c} \nuI{1} \\ \vdots 
                 \\ \nuI{g} \end{array}\right],
\end{equation*}
\begin{equation*}
v: S^k X \to J_X, \quad
  v((x_1,y_1), \cdots, (x_k, y_k))= \sum_{i=1}^k
       \int_\infty^{(x_i,y_i)} 
          \left[\begin{array}{c} \nuI{1}\\ \vdots \\ \nuI{g}
            \end{array}\right]
    \hbox{\rm modulo}\ {\Gamma_X}. 
\end{equation*}
The image of $v$ is denoted by $W_X^k =  v(S^k X)$, i.e., $J_X=W_X^g$.
Furthermore, we introduce an injection $\iota_X : X \to \tX$ and will fix it.
We find that $v = \kappa_J \circ \tv \circ \iota_X$.

Further, we introduce the differentials of the second kind,
$$
   \nuII{j}=\dfrac{1}{2y}\sum_{k=j}^{2g-j}(k+1-j)
      \lambda_{k+1+j} x^k d x,
     \quad (j=1, \ldots, g),
$$
and the complete hyperelliptic integrals of the second kind,
$$
\eta'_{i j}:=-\frac{1}{2}\int_{\alpha_j} \nuII{i},\quad
\eta''_{i j}:=-\frac{1}{2}\int_{\beta_j} \nuII{i}.
$$
For this basis of the $2g$-dimensional space of meromorphic differentials, 
the half-periods $\omega',\omega'',\eta',\eta''$  satisfy   
the generalized Legendre relation,
\begin{equation} 
{\mathfrak M}\left(\begin{array}{cc}0&-1_g\\1_g&0\end{array}\right) 
{\mathfrak M}^T= 
\frac{\ii\pi}{2}\left(\begin{array}{cc}0&-1_g\\1_g&0\end{array}\right),
\label{eq:Legrel}
\end{equation} 
where  $\mathfrak 
M=\left( 
\begin{array}{cc}\omega'&\omega''\\\eta'&\eta''\end{array} 
\right)$.
Let  $\tau={\omega'}^{-1}\omega''$. 
Further, we also introduce the branchpoint integral of the second kind, $\displaystyle{
\eta_{B_a}:=-\int_{\infty}^{B_a} \nuII{}}$ formally.
Precisely, it should be defined as $\displaystyle{
\eta_{B_a}:=-\frac{1}{2}\int_{\gamma_{B_a,B_a}} \nuII{}}$, where $\gamma_{B_a,B_a}$ is a closed path from $B_a$ to $B_a$ via the vicinity of $\infty$ such that the closed path contains $\infty$ in its inner side.

In this paper we use the following equivalent and unambiguous convention.
By Lemma \ref{2lm:PerodM_hyp2}, every branch-point integral is a half-period, i.e., there are integer vectors $\varepsilon_a'={}^t(\varepsilon_{a,1}',\ldots,\varepsilon_{a,g}')$, $\varepsilon_a''\in\ZZ^g$ such that
\begin{equation}
\omega_{B_a}=\sum_{i=1}^g\bigl(\varepsilon_{a,i}'\,\omega_i'+\varepsilon_{a,i}''\,\omega_i''\bigr),
\quad\mbox{and we set}\quad
\eta_{B_a}:=\sum_{i=1}^g\bigl(\varepsilon_{a,i}'\,\eta_i'+\varepsilon_{a,i}''\,\eta_i''\bigr),
\label{eq:etaBa}
\end{equation}
where $\eta_i'={}^t(\eta'_{1i},\ldots,\eta'_{gi})$ and $\eta_i''={}^t(\eta''_{1i},\ldots,\eta''_{gi})$.
In the notation of (\ref{eq:trans}) below, $\ell=2\omega_{B_a}$ has $\ell'=\varepsilon_a'$, $\ell''=\varepsilon_a''$, and $\eta_{B_a}=\frac12(2\,{}^t\eta'\ell'+2\,{}^t\eta''\ell'')$.
For instance, $\varepsilon_1'=0$, $\varepsilon_1''=\be_1$ and $\varepsilon_2'=-\be_1$, $\varepsilon_2''=\be_1$, where $\be_1={}^t(1,0,\ldots,0)$ (Corollary \ref{cor:omega}).
With (\ref{eq:etaBa}), $2\omega_{B_a}\in\Gamma_X$ and the translational formula (\ref{eq:trans}) can be applied to $\ell=2\omega_{B_a}$ directly.

\begin{definition}\label{4df:KdV_def2}
Let $(x_i, y_i)_{i=1, \ldots, g} \in S^g X$.
\begin{enumerate}
\item 
We define the polynomials associated with $F(x):=(x-x_1) \cdots (x-x_g)$ by
\begin{equation}
\pi_i(x) := \frac{F(x)}{x-x_i}=\chi_{i,g-1}x^{g-1} +\chi_{i,g-2} x^{g-2}
            +\cdots+\chi_{i,1}+\chi_{i,0},
\label{4eq:KdV_def2.1}
\end{equation}
\begin{equation}
F(x)=x^g+\varepsilon_{g-1} x^{g-1}+\cdots \varepsilon_1 x + \varepsilon_0.
\label{4eq:KdV_def2.1a}
\end{equation}

\end{enumerate}
We write $F'(x):=dF(x)/dx$.

\end{definition}

\begin{lemma}\label{4lm:KdV1}
Let $u \in \tv(\iota_X((x_1, y_1), \ldots, (x_g, y_g)))$.

The Dubrovin-type equations \cite{D, DT, GH, M25} hold:
$$
\frac{\partial x_i}{\partial u_r}=
\frac{2y_i}{F'(x_i)} \chi_{i, r-1}, \quad
\frac{\partial y_i}{\partial u_r}=
\frac{f'(x_i)}{F'(x_i)} \chi_{i,r-1},
$$
\begin{equation}
\frac{\partial}{\partial u_g }=
         \sum_{i=1}^g \frac{2y_i}{F'(x_i)} \frac{\partial}{\partial x_i},
           \quad
	\frac{\partial}{\partial u_{g-1} }=
         \sum_{i=1}^g \frac{2y_i\chi_{i,g-2}}{F'(x_i)}
              \frac{\partial}{\partial x_i}.
\label{4eq:hyp_dxdu}
\end{equation}
\end{lemma}

$F(x)$ is the same as $U$ in Mumford's $UVW$ expression \cite{Mum84}.


\subsection{Sigma function and its derivatives}

The theta function on  $\mathbb C^g$  with the modulus  $\tau:=\omega^{\prime -1}\omega''$ 
and characteristics $\delta', \delta''\in\displaystyle{
\left[\frac{1}{2} \ZZ\right]^g}$ is given by
\begin{equation}
    \theta\left[\begin{array}{c}\delta'' \\ \delta' \end{array}\right]
     (z; \tau)
    =\sum_{n \in \mathbb Z^g} \exp \left[\pi \ii\left\{
     \trp(n+\delta'')\tau (n+\delta'')
    +2 \trp(n+\delta'')(z+\delta')\right\}\right].
\label{eq:theta}
\end{equation}
The $\sigma$-function (\cite{B1, BEL, BEL20, M25}), 
an analytic function on the space  $\CC^g$  and a theta series having 
modular invariance with respect to  $\mathfrak{M}$, 
is given by the formula,
\begin{equation}
 \sigma(u)
  =\gamma_0\, \mathrm{exp}\left\{-\frac{1}{2}\ ^t u
  \eta'{{\omega}'}^{-1}u   \right\}
  \theta\negthinspace
  \left[\begin{array}{c}\delta'' \\ \delta' \end{array}\right]
  \left(\frac{1}{2}{{\omega}'}^{-1}u \,;\, \mathbb T\right),
\label{eq:sigma}
\end{equation}
where we use the convention \cite[3.82]{Mum84}
$$
\delta':=\left[\begin{matrix}\frac{g}{2}\\
\frac{g-1}{2}\\ \vdots\\ \frac{1}{2}\end{matrix}\right], \quad
\delta'':=\left[\begin{matrix}\frac{1}{2}\\
\frac{1}{2}\\ \vdots\\ \frac{1}{2}\end{matrix}\right],
$$
and $\gamma_0$ is a certain non-zero constant. 
The $\sigma$-function vanishes only on  $\kappa_J^{-1}(W_X^{g-1})$. 
The hyperelliptic $\wp$ and $\zeta$ functions are defined by
\begin{equation}
\wp_{i j} := -\frac{\partial^2}{\partial u_i \partial u_j} \log \sigma(u),\quad
\zeta_i = -\frac{\partial}{\partial u_i} \log \sigma(u).
\end{equation}

Let $u \in \tv(\iota_X((x_1, y_1), \ldots, (x_g, y_g)))$.
As the Jacobi inversion formula, we have the relation,
$$
F(x)=x^g -\sum_{i=1}^g \wp_{gi}(u) x^{i-1}, \quad
\mbox{e.g.,}\quad \wp_{gg}=x_1 + x_2 + \cdots + x_g.
$$

We introduce a multi-index  $\natural_n$.  
For  $n$  with  $1\leq n<g$, we let  $\natural_n$  
be the sequence of positive integers  $i$  such that  
$n+1\leq i\leq g$  with  $i\equiv n+1$ modulo $2$.   
In other words, we have 
\begin{equation*}
\natural_n=\left\{\begin{array}{lll}
n+1, n+3, \cdots, g-1 \quad& \mbox{for}\; g-n\equiv 0\; \mbox{modulo}\; 2,\\
n+1, n+3, \cdots, g   \quad& \mbox{for}\; g-n\equiv 1\; \mbox{modulo}\; 2,\\
 \end{array} \right.
\end{equation*}
and the partial derivative over the multi-index $\natural_n$,   
$$
	\sigma_{\natural_n}
        =\bigg(\prod_{i\in \natural_n}
	\frac{\partial}{\partial u_i}\bigg)  \sigma(u).
$$
For  $n\geq g$, we define  $\natural_n$  as empty and  $\sigma_{\natural_n}$  as  $\sigma$  itself.

 For  $u\in\CC^g$, we denote by  $u'$  and  $u''$  
the unique vectors in  $\mathbb{R}^g$  such that 
  $$
  u=2\,{}^t\omega'u'+2\,{}^t\omega'' u''.
  $$
Using the convention, we define 
  \begin{eqnarray*}
  L(u,v)&=&{}^t{u}(2\,{}^t\eta'v'+2\,{}^t\eta''v''), \\ 
  \chi(\ell)&=&
  \exp\big\{2\pi i\big({}^t{\ell'}\delta''-{}^t{\ell''}\delta'
  +\frac12{}^t{\ell'}\ell''\big)\big\}
  \ (\in \{1, -1\})
  \end{eqnarray*}
for  $u$, $v\in\CC^g$  and   $\ell$ 
($=2\,{}^t\omega'\ell'+2\,{}^t\omega''\ell''$) $\in\Gamma_X$.  
Then  $\sigma_{\natural_n}(u)$  for  $u\in\kappa_J^{-1}(W_X^n)$ 
 satisfies the translational relation \cite[Theorem 3.47]{M25}:
\begin{equation}
 \sigma_{\natural_n}(u+\ell)
   =\chi(\ell)\sigma_{\natural_n}(u)\exp L(u+\frac12\ell,\ell)\ \ 
 \hbox{for $u\in\kappa_J^{-1}(W_X^n)$}. 
\label{eq:trans}
\end{equation}

\begin{lemma}\label{lm:2.6}

Let $\omega_{12}:=\omega_{B_1}+\omega_{B_2}$ and $\eta_{12}:=\eta_{B_1}+\eta_{B_2}$.
For $u\in\kappa_J^{-1}(W_X^n)$ and $a=1,2$,
$$
 \sigma_{\natural_n}(u\pm2\omega_{B_a})
   =\chi(2\omega_{B_a})\,\ee^{\pm2\trp(u\pm\omega_{B_a})\eta_{B_a}}\sigma_{\natural_n}(u),
$$
$$
 \sigma_{\natural_n}(u\pm2\omega_{12})
   =\chi(2\omega_{12})\,\ee^{\pm2\trp(u\pm\omega_{12})\eta_{12}}\sigma_{\natural_n}(u),
$$
where, for the characteristics $\delta'$, $\delta''$ of (\ref{eq:sigma}),
$\chi(2\omega_{B_1})=\chi(2\omega_{B_2})=(-1)^g$ and $\chi(2\omega_{12})=-1$.

\end{lemma}

\begin{proof}
By Corollary \ref{cor:omega}, $\ell=2\omega_{B_1}$ has $\ell'=0$, $\ell''=\be_1$; $\ell=2\omega_{B_2}$ has $\ell'=-\be_1$, $\ell''=\be_1$; and $\ell=2\omega_{12}$ has $\ell'=-\be_1$, $\ell''=2\be_1$.
Since the first entry of $\delta'$ is $g/2$ and every entry of $\delta''$ is $1/2$, we obtain
$\chi(2\omega_{B_1})=\ee^{2\pi\ii(-g/2)}=(-1)^g$, $\chi(2\omega_{B_2})=\ee^{2\pi\ii(-1/2-g/2-1/2)}=(-1)^g$ and $\chi(2\omega_{12})=\ee^{2\pi\ii(-1/2-g-1)}=-1$;
moreover $\chi(-\ell)=\chi(\ell)$.
By (\ref{eq:etaBa}), $L(u\pm\omega_{B_a},\pm2\omega_{B_a})=\pm2\trp(u\pm\omega_{B_a})\eta_{B_a}$, and (\ref{eq:trans}) gives the claim; the case of $\omega_{12}$ is the same.
Note that $\chi$ cancels in all the multipliers of the $\al$ functions under $\Gamma_X$, see Lemma \ref{lm:multiplier}, but it enters the half-period relations in Lemmas \ref{lm:al_a}\,(6) and \ref{lm:al12_omega}\,(3).

\end{proof}

For  $n\leqq g$,  we note that 
$\sigma_{\natural_n}(-u)=(-1)^{ng+\frac12n(n-1)}\sigma_{\natural_n}(u)$ 
for $u\in \kappa_J^{-1}(W_X^n)$, especially \cite[Corollary 3.51]{M25}, 
\begin{equation}
\left\{
\begin{array}{llll}
\sigma_{\natural_1}(-u)=(-1)^g\sigma_{\natural_1}(u) \ \ 
        &\hbox{for $u\in \kappa_J^{-1}(W_X^1)$},  \\
\sigma_{\natural_2}(-u)=-\sigma_{\natural_2}(u) \ \ 
         & \hbox{for $u\in \kappa_J^{-1}(W_X^2)$}. 
\end{array}\right.
\label{eq:sign}
\end{equation}

\section{The \texorpdfstring{$\al_a$}{al\textunderscore a} and \texorpdfstring{$\al_{ab}$}{al\textunderscore ab} functions and their basic properties}

\subsection{The sigma function at half-periods}

The following relations are consequences of the addition formula of the hyperelliptic sigma function (\cite[Theorem 5.1]{EEMOP}, \cite[Section 3.6]{M25}) applied to one or two branch points; the second one is due to Bolza \cite{Bol95}.
Let $u=\tv(\iota_X(x_1,y_1),\ldots,\iota_X(x_g,y_g))$ and, for $x', x''\in\CC$,
\begin{equation}
\Delta_{x',x''}:=\sum_{r=1}^g\frac{y_r}{(x'-x_r)(x''-x_r)F'(x_r)}.
\label{eq:Delta}
\end{equation}

\begin{corollary}\label{3cr:al^2}
For $a=1,2$,
\begin{align*}
\left[\frac{\ee^{-\trp u\,\eta_{B_a}}
\sigma(u+ \omega_{B_a})}
{\sigma(u)\sigma_{\natural_1}(\omega_{B_a})}\right]^2
&= (-1)^gF(b_a)=\prod_{i=1}^g(x_i-b_a),\\
\left[\frac{\ee^{-\trp\omega_{B_1}\eta_{B_2}}
\sigma_{\natural_2}(\omega_{B_1}+ \omega_{B_2})
}
{\sigma_{\natural_1}(\omega_{B_1})\sigma_{\natural_1}(\omega_{B_2})}
\right]^2
&= b_1-b_2.
\end{align*}
\end{corollary}

\begin{corollary}\label{3cr:addhyp6b}
\begin{equation}
\left[\frac{\ee^{-\trp u\,
       (\eta_{B_1}+\eta_{B_2})}\sigma(u + \omega_{B_1}+\omega_{B_2}) 
}
{\sigma(u) \sigma_{\natural_2}(\omega_{B_1}+\omega_{B_2})}\right]^2
=\prod_{i=1}^g(x_i-b_1)(x_i-b_2)\,\Delta_{b_1b_2}^2.
\label{3eq:add_n_2II2b}
\end{equation}
\end{corollary}

\begin{proof}[Proof of Corollaries \ref{3cr:al^2} and \ref{3cr:addhyp6b}]
The addition formula with one branch point $v=\omega_{B_a}$ gives $\sigma(u+\omega_{B_a})\sigma(u-\omega_{B_a})/(\sigma(u)^2\sigma_{\natural_1}(\omega_{B_a})^2)=F(b_a)$ \cite[Corollary 3.76]{M25}, and with two branch points $v=\omega_{B_1}+\omega_{B_2}$ it gives $-F(b_1)F(b_2)\Delta_{b_1b_2}^2$ for the corresponding quotient with $\sigma_{\natural_2}$ \cite[Lemma 3.79]{M25}, \cite{Bol95}.
Then Lemma \ref{lm:2.6}, i.e., $\sigma(u-\omega)=\chi(2\omega)\ee^{-2\trp u\eta}\sigma(u+\omega)$ with $\chi(2\omega_{B_a})=(-1)^g$ and $\chi(2\omega_{12})=-1$, gives the claims.
\end{proof}

\begin{remark}\label{rmk:MumfordUVW1}
{\rm{
In Mumford's $UVW$-expression \cite[3.69]{Mum84}, $F=U$ and $F(x)\Delta_x=V(x)$ with $\displaystyle{\Delta_x:=\sum_r \frac{y_r}{(x-x_r)F'(x_r)}}$, so that $\displaystyle{\Delta_{x',x''}=\frac{U(x')V(x'')-U(x'')V(x')}{U(x')U(x'')(x'-x'')}}$, and $\al_a^2=(-1)^gU(b_a)$.
Thus the $\al$ functions are expressed by the variables of the Mumford dynamical system, which has recently been studied by means of the hyperelliptic sigma function \cite{Bu, Ba1, Ba2}.
}}
\end{remark}

\subsection{The definitions and basic properties of the \texorpdfstring{$\al_a$}{al\textunderscore a} and \texorpdfstring{$\al_{ab}$}{al\textunderscore ab} functions}

We will introduce the $\al_a$ and $\al_{12}$ functions and show their basic properties.

\begin{definition}\label{def:als}
We define $\al_1$, $\al_2$ and $\al_{12}$ by
\begin{equation}
\al_a(u):=\gamma_a''\frac{\ee^{- \trp u \eta_{B_a}}
       \sigma{}(u +  \omega_{B_a})}
{\sigma{}(u) \sigma_{\natural_1}(\omega_{B_a})},
\quad (a=1, 2),
\label{eq:def_ala}
\end{equation}
\begin{equation}
\al_{12}(u):=\gamma_{12}'' \frac{
\ee^{-\trp u  (\eta_{B_1}+\eta_{B_2})}
\sigma( u + \omega_{B_1}+ \omega_{B_2})}
{\sigma(u)\sigma_{\natural_2}(\omega_{B_1}+\omega_{B_2})},
\label{eq:def_al12}
\end{equation}
where $\gamma_1''=\gamma_2''=1$, and $\gamma_{12}''\in\{1,-1\}$ is the sign for which
\begin{equation}
\al_{12}(u)=\al_1(u)\al_2(u)\,\Delta_{b_1,b_2}(u)
\label{eq:al12norm}
\end{equation}
holds (see Lemma \ref{lm:Delta_ab}\,(1)).
\end{definition}

\begin{remark}\label{rmk:normalization}
{\rm By Corollaries \ref{3cr:al^2} and \ref{3cr:addhyp6b}, this normalization gives
$\al_a(u)^2=\prod_{i=1}^g(x_i-b_a)$ and $\al_{12}(u)^2=\prod_{i=1}^g(x_i-b_1)(x_i-b_2)\,\Delta_{b_1b_2}^2$.
Both sides of (\ref{eq:al12norm}) are meromorphic on $\CC^g$ and have the same square, so their ratio is a constant sign, which is absorbed in $\gamma_{12}''$.
}
\end{remark}

\begin{remark}\label{rmk:3.6} 
{\rm{
As we defined $\al_1$, $\al_2$ and $\al_{12}$ functions, it is easy to extend them to general $\al_a$ and $\al_{ab}$ ($a, b = 1, 2, \ldots, 2g+1$) for certain constant factors $\gamma_a''$ and $\gamma_{ab}''$  as mentioned in \cite {M25}.
Hence, our following results are of $\al_a$ and $\al_{ab}$ in this paper can be generalized to them.

Further, $\al_a$ and $\al_{12}$ are not meromorphic functions on the Jacobian $J_X$ although they are the Abelian functions as we show in Lemmas \ref{lm:al_a} and \ref{lm:al12_omega}.
By Lemma \ref{lm:multiplier} below, each of $\al_1$, $\al_2$ and $\al_{12}$ is single-valued on its own unramified double covering $\hJ_{X,1}$, $\hJ_{X,2}$, $\hJ_{X,12}$ of $J_X$ (Section \ref{sec:Wcover} and Subsection \ref{subs:geom_al12}).
The fourfold covering $\hJ_X=\CC^g/(\Gamma^{(1)}\cap\Gamma^{(2)})$ is the smallest covering on which the three functions are single-valued simultaneously, but it is not needed for $\al_{12}$ alone.
}}
\end{remark}

The behavior of the al functions under the lattice $\Gamma_X$ is described uniformly as follows.

\begin{lemma}\label{lm:multiplier}
Let $\omega_B=\sum_i(\varepsilon_i'\omega_i'+\varepsilon_i''\omega_i'')$ be a half-period with $\varepsilon',\varepsilon''\in\ZZ^g$, let $\eta_B$ be defined as in (\ref{eq:etaBa}), and let
$f_B(u):=\ee^{-\trp u\eta_B}\sigma(u+\omega_B)/\sigma(u)$.
Then for every $\ell=2\,{}^t\omega'\ell'+2\,{}^t\omega''\ell''\in\Gamma_X$,
$$
f_B(u+\ell)=(-1)^{\trp\varepsilon'\ell''+\trp\varepsilon''\ell'}f_B(u).
$$
Consequently, for $\ell\in\Gamma_X$,
\begin{equation}
\al_1(u+\ell)=(-1)^{\ell_1'}\al_1(u),\quad
\al_2(u+\ell)=(-1)^{\ell_1'+\ell_1''}\al_2(u),\quad
\al_{12}(u+\ell)=(-1)^{\ell_1''}\al_{12}(u).
\label{eq:multipliers}
\end{equation}
In particular, $\al_a$ $(a=1,2)$ and $\al_{12}$ are invariant exactly under the sublattices of index two
\begin{gather*}
\Gamma^{(1)}:=\{\ell\in\Gamma_X\,|\,\ell_1'\equiv0\ \mbox{modulo}\ 2\},\quad
\Gamma^{(2)}:=\{\ell\in\Gamma_X\,|\,\ell_1'+\ell_1''\equiv0\ \mbox{modulo}\ 2\},\\
\Gamma^{(12)}:=\{\ell\in\Gamma_X\,|\,\ell_1''\equiv0\ \mbox{modulo}\ 2\},
\end{gather*}
and $\Gamma^{(1)}\cap\Gamma^{(2)}=\Gamma^{(1)}\cap\Gamma^{(12)}=\Gamma^{(2)}\cap\Gamma^{(12)}$ has index four in $\Gamma_X$.
\end{lemma}

\begin{proof}
By (\ref{eq:trans}) the factors $\chi(\ell)$ cancel in $f_B(u+\ell)/f_B(u)$, and
$f_B(u+\ell)/f_B(u)=\exp\bigl(L(\omega_B,\ell)-\trp\ell\,\eta_B\bigr)=\exp\bigl(L(\omega_B,\ell)-L(\ell,\omega_B)\bigr)$.
The generalized Legendre relation (\ref{eq:Legrel}) shows $L(\omega_B,\ell)-L(\ell,\omega_B)\equiv\pi\ii(\trp\varepsilon'\ell''-\trp\varepsilon''\ell')$ modulo $2\pi\ii$.
The al functions are constant multiples of $f_B$ with $(\varepsilon',\varepsilon'')=(0,\be_1)$, $(-\be_1,\be_1)$ and $(-\be_1,2\be_1)$ for $B=B_1$, $B_2$ and $\omega_B=\omega_{12}$ (Corollary \ref{cor:omega}), which gives (\ref{eq:multipliers}).
\end{proof}

The $\al_a$ functions satisfy the following relations as their basic properties, which have never been reported anywhere:

\begin{lemma}\label{lm:al_a}
\begin{enumerate}
\item $\displaystyle{
\al_a(u)^2 = (-1)^g F(b_a)=\prod_{i} (x_i - b_a)
}$,

\item $\displaystyle{
\al_a(u+4\omega_{B_a})=\al_a(u)
}$,

\item $\displaystyle{
\al_1(u+2\omega_{B_1})=\al_1(u), \quad
\al_1(u+2\omega_{B_2})=-\al_1(u)
}$,

\item $\displaystyle{
\al_2(u+2\omega_{B_1})=-\al_2(u), \quad
\al_2(u+2\omega_{B_2})=\al_2(u)
}$,

\item $\al_a(u+\ell)=\al_a(u)$ for $\ell\in\Gamma^{(a)}$, see Lemma \ref{lm:multiplier},

\item $\displaystyle{
\al_a(u+\omega_{B_a})\,\al_a(u)=
\chi(2\omega_{B_a})\,(\gamma_a'')^2\,\frac{\ee^{\trp \omega_{B_a} \eta_{B_a}}}{\sigma_{\natural_1}(\omega_{B_a})^2}
=(-1)^g\frac{\ee^{\trp \omega_{B_a} \eta_{B_a}}}{\sigma_{\natural_1}(\omega_{B_a})^2}
}$ for $a=1,2$,

\item 
$\displaystyle{
\bigl(\al_a(u+\omega_{B_a})\,\al_a(u)\bigr)^2=f'(b_a)=\prod_{j\neq a}(b_a-b_j)
}$ for $a=1,2$,

\item 
$\displaystyle{
\al_1(u+\omega_{B_2})\al_2(u)=
\left[\frac{\gamma_1''\gamma_2''\ee^{- \trp \omega_{B_2} \eta_{B_1}}
\sigma_{\natural_2}(\omega_{B_1}+\omega_{B_2})
}{
\gamma_{12}''\sigma_{\natural_1}(\omega_{B_1})
\sigma_{\natural_1}(\omega_{B_2})}\right]\al_{12}(u)
}$.
\end{enumerate}
\end{lemma}

\begin{proof}
(1) is obvious. (2) is obtained by (3), (4) and (5).

(2)--(5) follow from Lemma \ref{lm:multiplier}.

(6): By Lemma \ref{lm:2.6}, $\sigma(u+2\omega_{B_a})=\chi(2\omega_{B_a})\ee^{2\trp(u+\omega_{B_a})\eta_{B_a}}\sigma(u)$, hence
\begin{align*}
\al_a(u+\omega_{B_a})\al_a(u)&=(\gamma_a'')^2
\frac{\ee^{-\trp(u+\omega_{B_a})\eta_{B_a}}\sigma(u+2\omega_{B_a})}{\sigma(u+\omega_{B_a})\sigma_{\natural_1}(\omega_{B_a})}
\frac{\ee^{-\trp u\eta_{B_a}}\sigma(u+\omega_{B_a})}{\sigma(u)\sigma_{\natural_1}(\omega_{B_a})}\\
&=\chi(2\omega_{B_a})(\gamma_a'')^2\frac{\ee^{\trp\omega_{B_a}\eta_{B_a}}}{\sigma_{\natural_1}(\omega_{B_a})^2}.
\end{align*}

(7): The translation by $\omega_{B_a}$ corresponds to $D\mapsto D+B_a-\infty$ for the divisor $D=\sum_i(x_i,y_i)$.
Let $V(x)$ be the polynomial of degree $\le g$ with $V(x_i)=y_i$ and $V(b_a)=0$.
The function $y-V(x)$ on $X$ vanishes on $D+B_a$, and its remaining zeros $E=\sum_i(x_i',V(x_i'))$ are given by $(x-b_a)F(x)\,U'(x)=f(x)-V(x)^2$ with $U'(x)=\prod_i(x-x_i')$.
Hence $D+B_a-\infty$ is linearly equivalent to $D'=\iota E=\sum_i(x_i',-V(x_i'))$, $\iota$ being the hyperelliptic involution (cf.\ \cite[IIIa]{Mum84}).
Dividing by $x-b_a$ and putting $x=b_a$, we obtain $F(b_a)U'(b_a)=f'(b_a)-2V(b_a)V'(b_a)=f'(b_a)$, i.e., $\al_a(u)^2\al_a(u+\omega_{B_a})^2=\prod_i(x_i-b_a)\prod_i(x_i'-b_a)=f'(b_a)$.

(8) holds since $\displaystyle{
\al_1(u+\omega_{B_2})\al_2(u)=
\gamma_1''\gamma_2''\frac{\ee^{- \trp (u+\omega_{B_2}) \eta_{B_1}}
       \sigma{}(u +  \omega_{B_1}+\omega_{B_2})}
{\sigma{}(u+\omega_{B_2}) \sigma_{\natural_1}(\omega_{B_1})}
}\times$
$\displaystyle{
\frac{\ee^{- \trp u \eta_{B_2}}
       \sigma{}(u +  \omega_{B_2})}
{\sigma{}(u) \sigma_{\natural_1}(\omega_{B_2})}
}$.
\end{proof}

Similarly, we have the basic properties of the $\al_{12}$ functions, which have neither been reported anywhere:

\begin{lemma} \label{lm:al12_omega}
\begin{enumerate}
\item $\displaystyle{
\al_{12}(u)^2 =F(b_1)F(b_2)
\left[ \sum_{r=1}^g\frac{y_r}{(b_1 - x_r)(b_2 - x_r)F'(x_r)}\right]^2}$
$\displaystyle{
=F(b_1)F(b_2)\Delta_{b_1, b_2}^2
}$,

\item $\displaystyle{
\al_{12}(u+\omega_{B_1})=c_{12,1}\frac{\al_2(u)}{\al_1(u)}}$ (cf.\ \cite[(17)]{Wei54}), where $\displaystyle{c_{12,1}^2=\frac{f'(b_1)}{b_1-b_2}}$ and
$$
c_{12,1}:=\chi(2\omega_{B_1})\frac{\gamma_{12}''\gamma_{1}''}{\gamma_2''}
\frac{\ee^{\trp\omega_{B_1}\eta_{B_1}-\trp\omega_{B_1}\eta_{B_2}+2\trp\omega_{B_2}\eta_{B_1}}\sigma_{\natural_1}(\omega_{B_2})}
{\sigma_{\natural_2}(\omega_{B_1}+\omega_{B_2})\sigma_{\natural_1}(\omega_{B_1})},
$$

\item with $\omega_{12}=\omega_{B_1}+\omega_{B_2}$, $\eta_{12}=\eta_{B_1}+\eta_{B_2}$,
$$
\al_{12}(u+\omega_{12})\,\al_{12}(u)=\chi(2\omega_{12})(\gamma_{12}'')^2\frac{\ee^{\trp\omega_{12}\eta_{12}}}{\sigma_{\natural_2}(\omega_{12})^2}
=-\frac{\ee^{\trp\omega_{12}\eta_{12}}}{\sigma_{\natural_2}(\omega_{12})^2},
$$
$$
\bigl(\al_{12}(u+\omega_{12})\al_{12}(u)\bigr)^2=-\frac{f'(b_1)f'(b_2)}{(b_1-b_2)^2},
$$

\item $\al_{12}(u+\ell)=\al_{12}(u)$ for $\ell\in\Gamma^{(12)}$, see Lemma \ref{lm:multiplier}.

\item $\al_{12}(u+2\omega_{B_a})=-\al_{12}(u)$, $a=1,2$.
\end{enumerate}
\end{lemma}

\begin{proof}
(1) is obvious. (2): By Lemma \ref{lm:2.6}, $\sigma(u+2\omega_{B_1}+\omega_{B_2})=\chi(2\omega_{B_1})\ee^{2\trp(u+\omega_{B_1}+\omega_{B_2})\eta_{B_1}}\sigma(u+\omega_{B_2})$; inserting this and the definitions of $\al_1$, $\al_2$ gives $c_{12,1}$.
The value of $c_{12,1}^2$ follows from (\ref{eq:al12norm}), Lemma \ref{lm:al_a}\,(7) and Lemma \ref{lm:Delta_ab}\,(4), (5).
(3): As in the proof of Lemma \ref{lm:al_a}\,(6), Lemma \ref{lm:2.6} gives
$$
\al_{12}(u+\omega_{12})\al_{12}(u)=(\gamma_{12}'')^2\frac{\ee^{-\trp(2u+\omega_{12})\eta_{12}}\sigma(u+2\omega_{12})}{\sigma(u)\sigma_{\natural_2}(\omega_{12})^2}
=\chi(2\omega_{12})(\gamma_{12}'')^2\frac{\ee^{\trp\omega_{12}\eta_{12}}}{\sigma_{\natural_2}(\omega_{12})^2},
$$
with $\chi(2\omega_{12})=-1$.
The value of the square is obtained by the function $y-V(x)$ with $V(x_i)=y_i$ and $V(b_1)=V(b_2)=0$, in the same way as in the proof of Lemma \ref{lm:al_a}\,(7).
(4) and (5) follow from Lemma \ref{lm:multiplier}.
\end{proof}

Due to Lemma \ref{lm:al12_omega} 1 and Proposition \ref{3pr:Bol_rel}, the part $\Delta_{b_1,b_2}$ (\ref{eq:Delta}) of $\al_{12}$ has the following relations.

\begin{lemma}\label{lm:Delta_ab}
\begin{enumerate}
\item $\displaystyle{
\frac{\al_{12}(u)}{\al_1(u)\al_2(u)}=
\frac{\gamma_{12}''}{\gamma_1''\gamma_2''}
 \frac{
\sigma( u + \omega_{B_1}+ \omega_{B_2})\sigma(u)\sigma_{\natural_1}(\omega_{B_1})\sigma_{\natural_1}(\omega_{B_2})
}
{\sigma(u+\omega_{B_1})\sigma(u+\omega_{B_2})\sigma_{\natural_2}(\omega_{B_1}+\omega_{B_2})}}$
$\displaystyle{
=\Delta_{b_1,b_2}(u),
}$

\item $\displaystyle{
\Delta_{b_1,b_2}(u+ \omega_{B_1}+ \omega_{B_2}) = -\Delta_{b_1,b_2}(u),
}$

\item $\displaystyle{
\Delta_{b_1,b_2}(u)=
\frac{\gamma_{12}''}{\gamma_1''\gamma_2''}
 \frac{
\sigma_{\natural_1}(\omega_{B_1})\sigma_{\natural_1}(\omega_{B_2})
}{\ee^{-\trp \omega_{B_2}\eta_{B_1}}
\sigma_{\natural_2}(\omega_{B_1}+\omega_{B_2})}
\frac{\al_1(u+ \omega_{B_2})}{\al_1(u)}
}$,

\item $\displaystyle{
\Delta_{b_1,b_2}(u)=
\frac{\gamma_{12}''}{\gamma_1''\gamma_2''}
 \frac{
\sigma_{\natural_1}(\omega_{B_1})\sigma_{\natural_1}(\omega_{B_2})
}{\ee^{-\trp \omega_{B_1}\eta_{B_2}}\sigma_{\natural_2}(\omega_{B_1}+\omega_{B_2})}
\frac{\al_2(u+ \omega_{B_1})}
{\al_2(u)}
}$,

\item $\displaystyle{
\Delta_{b_1,b_2}(u+ \omega_{B_1})=
-\left[\frac{\gamma_{12}''}{\gamma_1''\gamma_2''}
 \frac{
\sigma_{\natural_1}(\omega_{B_1})\sigma_{\natural_1}(\omega_{B_2})
}{\ee^{-\trp \omega_{B_1}\eta_{B_2}}\sigma_{\natural_2}(\omega_{B_1}+\omega_{B_2})}\right]^2
\frac{1}{\Delta_{b_1,b_2}(u)}}$
$\displaystyle{
=\frac{1}{b_2-b_1}\,\frac{1}{\Delta_{b_1,b_2}(u)}
}$.
\end{enumerate}
\end{lemma}

\begin{proof}
(1) is obvious.

(2): By Lemma \ref{lm:2.6},
$\sigma(u+2\omega_{12})=\chi(2\omega_{12})\ee^{2\trp(u+\omega_{12})\eta_{12}}\sigma(u)$,
$\sigma(u+\omega_{B_2}+2\omega_{B_1})=\chi(2\omega_{B_1})\ee^{2\trp(u+\omega_{12})\eta_{B_1}}\sigma(u+\omega_{B_2})$ and
$\sigma(u+\omega_{B_1}+2\omega_{B_2})=\chi(2\omega_{B_2})\ee^{2\trp(u+\omega_{12})\eta_{B_2}}\sigma(u+\omega_{B_1})$.
Inserting them into (1) with $u$ replaced by $u+\omega_{12}$, the exponential factors cancel and we obtain
$\Delta_{b_1,b_2}(u+\omega_{12})=\chi(2\omega_{12})\chi(2\omega_{B_1})\chi(2\omega_{B_2})\Delta_{b_1,b_2}(u)=-\Delta_{b_1,b_2}(u)$.

(3) and (4) are obtained by Lemma \ref{lm:al_a}\,(8) and its analogue $\displaystyle{\al_2(u+\omega_{B_1})\al_1(u)=}$\break
$\displaystyle{\frac{\gamma_1''\gamma_2''\ee^{-\trp\omega_{B_1}\eta_{B_2}}\sigma_{\natural_2}(\omega_{12})}{\gamma_{12}''\sigma_{\natural_1}(\omega_{B_1})\sigma_{\natural_1}(\omega_{B_2})}\,\al_{12}(u)}$.

(5): By (4), $\Delta_{b_1,b_2}(u)\Delta_{b_1,b_2}(u+\omega_{B_1})=K^2\,\al_2(u+2\omega_{B_1})/\al_2(u)$, where $K$ is the constant in (4), and $\al_2(u+2\omega_{B_1})=-\al_2(u)$ by Lemma \ref{lm:multiplier}.
The value $1/(b_2-b_1)$ is obtained algebraically by the function $y-V(x)$ as in the proof of Lemma \ref{lm:al_a}\,(7).

\end{proof}

\section{The algebraic and geometric foundations of \texorpdfstring{$\al_a$}{al\textunderscore a} functions, and their differential identities}\label{sec:ala}

In this section, we review the algebraic and geometric foundations of $\al_a$ functions (\ref{eq:def_ala}), and their differential identities following \cite{M25} and the Appendix in \cite{MP15}.

\subsection{Double covering of \texorpdfstring{$X$}{X}}\label{sec:Wcover}

\begin{figure}
\begin{center}

\includegraphics[width=0.6\hsize]{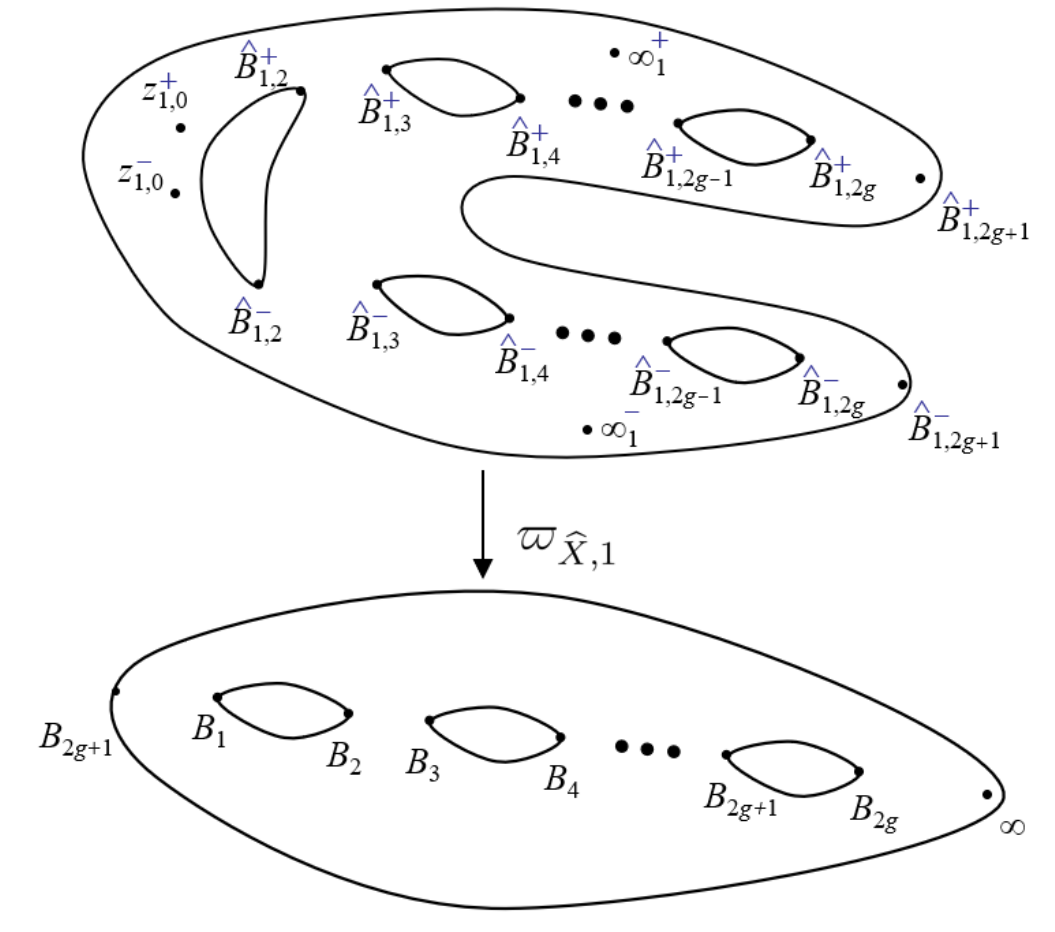}

\end{center}

\caption{The double covering $\varpi_{\hX,1}: \hX_1 \to X$,
$\varpi_X: (w_1, z_1) \mapsto (w_1^2+b_1, z_1w_1)=(x,y)$.
}\label{fg:Fig00}
\end{figure}

\bigskip

For a branch point $B_a=(b_a, 0) \in X$ ($\varpi_x: X\to \PP^1$) and $((x_i, y_i))_i \in S^g X$, in order to express $\al_a$ function (\ref{eq:def_ala}), $\al_a(u)=\sqrt{\prod_{i=1}^g(x_i-b_a)}$ (due to Lemma \ref{lm:al_a})  precisely, we introduce the double covering $\hX_a$ as in Figure \ref{fg:Fig00} and express it as in (\ref{eq:al_a=www}) in Proposition \ref{ala=www}.

	Although from the definition (\ref{eq:def_ala}), there is no ambiguity of the square root, we will give the expression as a meromorphic function on $\hX_a$ as the meaning of the Jacobi inversion formula.

The square root leads the transformation of $w_a^2 = (x-b_a)$, i.e., the double covering $\hX_a$ of the curve $X$, $\varpi_{\hX,a}: \hX_a \to X$, although the precise arguments are left to the Appendix in \cite{MP15}.

Since $\tX$ can be also an Abelian covering of $\hX_a$, we have a commutative diagram,
\begin{equation}
\xymatrix{ 
 \tX \ar[dr]^{\kappa_X}\ar[r]^-{\kappa_{\hX,a}}& \hX_a
 \ar[d]^-{\varpi_{\hX,a}} \\
  & X.
}\label{2eq:al_hyp_kappas}
\end{equation}

The curve $\hX_a$ is given by 
\begin{equation}
f_{\hX,a}(w_a,z_a) = z_a^2 -(w_a^2-e_{a,\overline{a}})(w_a^2-e_{a,3})\cdots(w_a^2 - e_{a,2g+1})=0,
\label{eq:fhXa}
\end{equation}
where $z_a:=y/w_a$ (due to the normalization \cite[Section 2.10]{M25}), $e_{a,j} := b_j- b_a$, $j=1, \ldots, 2g+1, (j \neq a)$, and 
 $\overline{a}=2$ for $a=1$ and $=1$ for $a=2$.
The projection $\varpi_{\hX,a}$ is given by $\varpi_{\hX,a}(w_a, z_a) =(w_a^2+b_a, z_a w_a)$ by noting $\varpi_{\hX,a}^{-1}\{\infty\}=\{\infty_a^+, \infty_a^-\}$ and $\varpi_{\hX,a}^{-1}\{(B_a,0)\} =\{ (0, z_{a,0}),   (0, -z_{a,0})\}$, where $z_{a,0}^2 = e_{a,\overline{a}}e_{a,3}\cdots e_{a,2g+1}$.
Its affine ring is 
\begin{equation}
R_{\hX,a}:=\CC[w_a,z_a]/(f_{\hX,a}(w_a, z_a)),
\label{eq:RhXa}
\end{equation}
and the ring of its $g$-th symmetric polynomials is denoted by $S^g R_{\hX,a}$.
We note that $R_{\hX,a}$ is the integral closure of $R_X[w_a]/(w_a^2-x+b_a)$ in its field of fractions: the latter ring has an ordinary double point at $w_a=y=0$ (since $y^2=w_a^2\prod_{j\ne a}(w_a^2-e_{a,j})$), and adjoining $z_a=y/w_a$ resolves it.
The covering $\hX_a\to X$ is unramified because the divisor of $x-b_a$ is $2B_a-2\infty$; it is the unramified double covering associated with the $2$-torsion point $v(B_a)$ of $J_X$.

Then we have the following proposition:
\begin{proposition}\label{ala=www}
For a point $(\gamma_1, \ldots, \gamma_g) \in S^g \tX$ satisfying 
$\kappa_{\hX,a}(\gamma_i) = (w_{a,i}, z_{a,i})$, and $u=\tv(\gamma_1, \ldots, \gamma_g)$, we have the equality of $\al$ function,
\begin{equation}
\al_a(u) =  w_{a,1}w_{a,2}\cdots w_{a,g} \in S^g R_{\hX,a}
\label{eq:al_a=www}
\end{equation}
as a meromorphic function on $S^g \hX_a$.
\end{proposition}
\begin{proof} 
See Appendix in \cite{MP15}.
Noting (\ref{2eq:al_hyp_kappas}), Lemma \ref{lm:al_a} also shows the fact. \end{proof}

\begin{remark} \label{rmk:al1=xxx}
{\rm{
Due to the double covering $\hX_a$ and the above proposition, the ambiguity in the root expression of $\al_a$ disappears. 
Furthermore, (\ref{eq:al_a=www}) and the definition of $\al_a$ provide the expressions of the meromorphic function on $S^g \hX_a$ in terms of the sigma function.
Thus this expression can be regarded as the Jacobi inversion formula for the double covering $\hX_a$.
}}
\end{remark}

Since $\hX_a$ is a unramified double covering of $X$,  the genus $g_{\hX_a}$ of $\hX$ is $2g-1$ given by the Hurwitz relations \cite{Kunz,M25},
\begin{gather}
2g_{\hX_a} - 2 = d_{\hX_a/X}(2g-2),
\label{eq:Hruwitz_hXa1}
\end{gather}
where $d_{\hX_a/X}=2$.


Thus, we have $2g-1$ holomorphic one-forms,
\begin{equation}
\hnu_{a,j}:= \frac{w_a^{j} d w_a}{z_a}, \quad (j=0, 1, 2, 3, \ldots, 2g-2),
\label{eq:hnu_aj}
\end{equation}
and the Jacobi variety, $J_{\hX,a}$ of $\hX_a$ is given by the complex torus $J_{\hX,a}=\CC^{2g-1} /\Gamma_{\hX,a}$ for the lattice $\Gamma_{\hX,a}$ generated by the period matrix $\Gamma_{\hX,a}$.

Since $x=w_a^2+b_a$, $dx=2w_adw_a$ and $y=z_aw_a$, we have
$\varpi_{\hX,a}^*\nuI{i}=(w_a^2+b_a)^{i-1}dw_a/z_a$, $(i=1,\ldots,g)$; hence the pull-backs of $\nuI{1},\ldots,\nuI{g}$ span the $g$-dimensional space spanned by the \emph{even} forms $\hnu_{a,0},\hnu_{a,2},\ldots,\hnu_{a,2g-2}$, i.e., the forms invariant under the covering involution $(w_a,z_a)\mapsto(-w_a,-z_a)$ (the odd forms span the $(g-1)$-dimensional space of Prym differentials) \cite[Appendix]{MP15}.

Integrating the pull-backs over the cycles of $\hX_a$ gives the lattice
$$
\Gamma^{(a)}:=\Bigl\{\oint_{\hgamma}\varpi_{\hX,a}^*\nuI{}\ \Big|\ \hgamma\in\mathrm{H}_1(\hX_a,\ZZ)\Bigr\}
=\varpi_{\hX,a*}\mathrm{H}_1(\hX_a,\ZZ)\subset\Gamma_X ,
$$
which is the kernel of the monodromy character $\mathrm{H}_1(X,\ZZ)\to\{\pm1\}$ of $w_a=\sqrt{x-b_a}$; it has index two in $\Gamma_X$ and coincides with $\Gamma^{(a)}$ of Lemma \ref{lm:multiplier} (e.g. for $a=1$ the cycle $\alpha_1$ encircles $b_1$ once and $w_1$ changes its sign, whereas the other basis cycles do not change $w_1$).
We set
$$
\hJ_{X,a}:=\CC^g/\Gamma^{(a)},\qquad \hkappa_{J,a}:\CC^g\to\hJ_{X,a},\qquad
\hvarpi_{J,a}:\hJ_{X,a}\to J_X ,
$$
which is an unramified double covering of $J_X$.
The projection of $\CC^{2g-1}$ onto the coordinates of the even forms induces a surjection $J_{\hX,a}\to\hJ_{X,a}$, and its composition with $\hvarpi_{J,a}$ is the norm map $\mathrm{Nm}:J_{\hX,a}\to J_X$, $\sum_i\hP_i\mapsto\sum_i\varpi_{\hX,a}(\hP_i)$.
Thus $\hJ_{X,a}$ is a quotient of $J_{\hX,a}$ which covers $J_X$.
In contrast, the pull-back $\varpi_{\hX,a}^*:J_X\to J_{\hX,a}$, $D\mapsto\varpi_{\hX,a}^{-1}(D)$, has the kernel $\{0,v(B_a)\}$, and its image $J_X/\langle v(B_a)\rangle$ is a subtorus of $J_{\hX,a}$ which is covered by $J_X$; we have $\mathrm{Nm}\circ\varpi_{\hX,a}^*=2$ on $J_X$.
The function $\al_a$ lives on the quotient $\hJ_{X,a}$, not on this subtorus.

For each branch point $B_j:=(b_j, 0)\in X$ $(j=1, \ldots, 2g+1, j \neq a)$, we have double branch points $\hB^\pm_{a,j}:=(\pm\sqrt{e_{a,j}},0) \in \hX_a$ as illustrated in Figure~\ref{fg:Fig00}.

Similar to the Jacobi elliptic functions, $\hJ_{X,a}=\CC^g/\Gamma^{(a)}$ is determined by the same Abelian integral $\tv$, and thus we use the same symbol $\tv$ as $\tv : S^g\tX \to \CC^g$ for $\hX_a$ \cite{MP15}.

It corresponds to the periodicity of the $\al_a$ function in Lemmas \ref{lm:al_a} and \ref{lm:multiplier}.

\bigskip

\subsection{Differential identities of the \texorpdfstring{$\al_a$}{al\textunderscore a} functions}

We show the differential identities of $\al_a$ functions following \cite[Section 4.2.2]{M25}.

Let $\displaystyle{\phi^{[a]}:=\log \al_a=\frac{1}{2}\log F(b_a)}$ for the hyperelliptic curves $X$ of genus $g$.

\begin{lemma}\label{lm:al_wp_phi}
\begin{equation}
\partial_{u_{g}} \phi^{[a]}
= 
\sum_{i=1}^g\frac{y_i}{F'(x_i) (x_i-b_a)}, \quad
	\partial_{u_{g-1}} \phi^{[a]}
         =\sum_{i=1}^g\frac{y_i\chi_{i,g-2}}{F'(x_i)( x_i-b_a)}.
\label{4eq:MKdV011aa}
\end{equation}
\end{lemma}

\begin{proof}
By applying (\ref{4eq:hyp_dxdu}) to them, we easily obtain them.
\end{proof}

\begin{theorem}\label{4th:al_wp_phi}
Let $\dot{h}:=\partial_{u_{g-1}}h$ and 
$h':=\partial_{u_{g}}h$.
The $\al_a(u)$ function holds the following relations:
\begin{enumerate}

\item 
$\displaystyle{
\partial_{u_{g}} \al_a = 
\left(
\partial_{u_{g}}\phi^{[a]}\right)
\al_a 
}$, \ i.e.,
$\displaystyle{
\partial_{u_{g}} \al_a = 
\left(\phi^{[a]\prime}\right)
\al_a 
}$,

\item 
$\displaystyle{
\partial_{u_{g}}^2\phi^{[a]}
+\left(\partial_{u_{g}}\phi^{[a]}
\right)^{\!\!2}
\!\!\!=2(x_1+x_2+\cdots+x_g) + \lambda_{2g} + b_a
=2\wp_{gg} + \lambda_{2g} + b_a}$,

\item 
$\displaystyle{
\partial_{u_{g}}^2 \al_a = 
\left(
\partial_{u_{g}}^2\phi^{[a]}
+\left(\partial_{u_{g}}\phi^{[a]}
\right)^2\right)
\al_a 
}$,  i.e., 
$\displaystyle{
\partial_{u_{g}}^2 \al_a -
2\wp_{g g}(u) \al_a =
(\lambda_{2g} + b_a)
\al_a
}$,

\item 
$
\displaystyle{
\dot{\phi}^{[a]}=\frac{1}{2}\wp_{ggg}-(\wp_{gg}-b_a)\phi^{[a]\prime}, 
}$

\item 
$
\displaystyle{
\dot{\phi}^{[a]}
=\frac14\phi^{[a]\prime\prime\prime}
-\frac12\phi^{[a]\prime3}
+\frac12(\lambda_{2g}+3b_a)\phi^{[a]\prime}.
}$

\end{enumerate}

\end{theorem}

\begin{proof}
(1) is trivial.
(2) is obtained in \cite{Mat02c} as the Miura transformation.
Then (3) is obvious.
(4) is obtained by the algebraic trick $\chi_{i, g-2}=x_i-\wp_{gg}$ (indeed $\chi_{i,g-2}=x_i+\varepsilon_{g-1}$ and $\varepsilon_{g-1}=-\wp_{gg}$), which gives
$\displaystyle{\dot\phi^{[a]}=\sum_i\frac{y_i(x_i-b_a+b_a-\wp_{gg})}{F'(x_i)(x_i-b_a)}=\frac12\wp_{ggg}-(\wp_{gg}-b_a)\phi^{[a]\prime}}$.
Then (5) is obtained by (2) and (4) as a hyperelliptic solution of the MKdV equation over $\CC$.
\end{proof}

\begin{remark}\label{rmk:eigen}
{\rm{
Theorem \ref{4th:al_wp_phi}\,(3) states that $\al_1$ and $\al_2$ are eigenfunctions of the Schr\"odinger operator $L:=\partial_{u_g}^2-2\wp_{gg}(u)$ with the eigenvalues $\lambda_{2g}+b_1$ and $\lambda_{2g}+b_2$; they are the classical eigenfunctions at the edges of the spectral bands of the finite-gap potential $2\wp_{gg}$ \cite{D, DT, GH, BBEIM}.
This is the background of the relations in Section 5.
}}
\end{remark}

\begin{remark}\label{rmk:al_aNLS_CMKdV}
{\rm{

From Theorem \ref{4th:al_wp_phi} (5), we have
\begin{equation}
\partial_{u_{g-1}}{\phi}^{[a]}
=\frac14\partial_{u_{g}}^2\phi^{[a]\prime}
-\frac12\phi^{[a]\prime3}
+\frac12(\lambda_{2g}+3b_a)\phi^{[a]\prime},
\label{eq:NLS_ala}
\end{equation}
\begin{equation}
[-4\partial_{u_{g-1}}+2(\lambda_{2g}+3b_a)\partial_{u_{g}}]
{\phi}^{[a]\prime}
-6\phi^{[a]\prime2}\partial_{u_{g}}\phi^{[a]\prime}
+\partial_{u_{g}}^3\phi^{[a]\prime}=0.
\label{eq:CMKdV_ala}
\end{equation}
Hence, for $\kappa^{[a]}:=-2\ii\,\phi^{[a]\prime}$ and $\partial_t:=-4\partial_{u_{g-1}}+2(\lambda_{2g}+3b_a)\partial_{u_g}$, $s:=u_g$, we obtain the focusing MKdV equation over $\CC$,
$$
\partial_t\kappa^{[a]}+\frac32(\kappa^{[a]})^2\partial_s\kappa^{[a]}+\partial_s^3\kappa^{[a]}=0,
$$
which is the form (\ref{4eq:MKdV_k}) studied in \cite{Mat02c, Ma24b, Ma26}.

}}
\end{remark}

\section{The algebraic and geometric foundations of \texorpdfstring{$\al_{12}$}{al\textunderscore 12} functions, and their differential identities}

In this section, we investigate the $\al_{12}$ function by recalling (\ref{eq:def_al12}) in Definition \ref{def:als}, Corollary \ref{3cr:addhyp6b}, Lemmas \ref{lm:al12_omega} and \ref{lm:Delta_ab} and by using the relations of $\al_a$ functions in the previous section.
The $\al_{ab}$ function goes back to Weierstrass's $\mathrm{al}(u)_{ab}$ \cite[(7)]{Wei54}; it was revisited in \cite[Section 4.2]{M25} following Baker and Bolza \cite{Baker98, Bol95}.

\begin{lemma} \label{lm:divisor}
$\omega_{B_a}+\omega_{B_b}$ is in the half lattice $\frac{1}{2}\Gamma_X$.
\end{lemma}

\begin{proof}
This follows from $2\omega_{B_a}\in\Gamma_X$ (Lemma \ref{2lm:PerodM_hyp2}); the divisor of $(x-b_a)(x-b_b)$ is $2(B_a+B_b)-4\infty$, so that $(x-b_a)(x-b_b)$ has a square root on an unramified double covering of $X$ (Subsection \ref{subs:geom_al12}).
\end{proof}

Since we also have $2(B_{a_1} + \cdots + B_{a_k}) - 2k\infty$ for $k \le g$, similarly, by extending the proof of Lemma \ref{lm:divisor}, we can define $\al_{a_1\cdots a_k}$ functions as in Baker's study \cite{Baker98}.

\bigskip

Lemmas \ref{lm:al12_omega} and \ref{lm:al_wp_phi} lead the following relations between $\al_{12}$ and $\al_a$ functions.

\begin{proposition}[{\cite[(7)]{Wei54}, \cite{Bol95}, \cite[Section IX]{Baker98}}]\label{3pr:Bol_rel}
\begin{gather}
\begin{split}
\al_{12}(u)
&=
\al_1(u)\al_2(u)\sum_{r=1}^g\frac{y_r}{(b_1 - x_r)(b_2 - x_r)F'(x_r)}\\
&
=
\al_1(u)\al_2(u)\Delta_{b_1, b_2}
\\
&=\frac{1}{b_2-b_1}
\left(\al_1
\partial_{u_{g}} \al^{[2]}
-\al_2\partial_{u_{g}} \al^{[1]}\right).
\end{split}
\label{eq:al12Delta}
\end{gather}
\end{proposition}

\begin{proof}
Recalling (\ref{4eq:MKdV011aa}), we have
$$
\frac{1}{b_b-b_a}
\left[\al_a
\partial_{u_{g}} \al_b
-\al_b\partial_{u_{g}} \al_a\right]
=\al_a\al_b\frac{\phi^{[b]\prime}-\phi^{[a]\prime}}{b_b-b_a}
=\al_a\al_b
\sum_{i=1}^g \frac{y_i}{F'(x_i) (b_a-x_i)(b_b-x_i)}.
$$
Here we use the relation
$\displaystyle{
\sum_{i=1}^g \left[\frac{y_i}{F'(x_i) (b_a-x_i)}
-\frac{y_i}{F'(x_i) (b_b-x_i)}\right]=
\sum_{i=1}^g \frac{y_i(b_b-b_a)}{F'(x_i) (b_a-x_i)(b_b-x_i)}
}$.
The first equality in (\ref{eq:al12Delta}) is the normalization (\ref{eq:al12norm}).
\end{proof}

\begin{remark}\label{rmk:Crum}
{\rm{
The last expression in (\ref{eq:al12Delta}) says that
$$
\al_{12}=\frac{1}{b_2-b_1}\,W(\al_1,\al_2),\qquad W(\psi_1,\psi_2):=\psi_1\,\partial_{u_g}\psi_2-\psi_2\,\partial_{u_g}\psi_1,
$$
is the Wronskian of the two eigenfunctions $\al_1$, $\al_2$ of $L=\partial_{u_g}^2-2\wp_{gg}$ (Remark \ref{rmk:eigen}); Theorem \ref{4th:alab_phi}\,(2) below, $\partial_{u_g}\al_{12}=\al_1\al_2$, is the classical identity $\partial W(\psi_1,\psi_2)=(E_2-E_1)\psi_1\psi_2$ for eigenfunctions with eigenvalues $E_1$, $E_2$.
Since $\log\al_{12}=-\trp u\,\eta_{12}+\log\sigma(u+\omega_{12})-\log\sigma(u)+\mathrm{const.}$, Crum's formula for the iterated Darboux transformation \cite{Crum} (see also \cite{Matveev79, MS91}) gives
$$
2\wp_{gg}(u)-2\partial_{u_g}^2\log W(\al_1,\al_2)=2\wp_{gg}(u+\omega_{B_1}+\omega_{B_2}),
$$
and similarly $2\wp_{gg}(u)-2\partial_{u_g}^2\log\al_a=2\wp_{gg}(u+\omega_{B_a})$.
Thus the single and double Darboux transformations at the band edges $b_1$, $b_2$ are realized by the half-period translations $u\mapsto u+\omega_{B_a}$ and $u\mapsto u+\omega_{B_1}+\omega_{B_2}$, and $\al_a$, $\al_{12}$ are the corresponding (gauge-transformed) $\tau$-function ratios.
This is the structural reason for the identities in Subsections 5.2 and 5.3.
}}
\end{remark}

Baker \cite[Section IX]{Baker98} expressed the shifts $\sigma(u+\omega_{B_{a_1}}+\cdots+\omega_{B_{a_k}})$ by determinants of the functions shifted by one and two half-periods.
Let $\al_{a_1\cdots a_k}$ denote the function defined by $\sigma(u+\omega_{B_{a_1}}+\cdots+\omega_{B_{a_k}})/\sigma(u)$ with exponential and constant factors as in Definition \ref{def:als}.
For $k=3$, Baker's formula \cite[p.~354]{Baker98} reads
\begin{equation}
\al_{123}=c'\,\frac{\al_2\al_{13}-\al_1\al_{23}}{b_1-b_2},\qquad c'\in\CC^\times .
\label{eq:Baker123}
\end{equation}

\begin{proposition}\label{pr:Baker_Crum}
The function $\al_{123}$ is Crum's Wronskian of the three band-edge eigenfunctions $\al_1$, $\al_2$, $\al_3$:
$$
\al_{123}=c\,\frac{W(\al_1,\al_2,\al_3)}{(b_2-b_1)(b_3-b_1)(b_3-b_2)},\qquad c\in\CC^\times .
$$
\end{proposition}

\begin{proof}
Let $W_{ab}=W(\al_a,\al_b)=(b_b-b_a)\al_{ab}$ (Remark \ref{rmk:Crum}).
By Remark \ref{rmk:eigen}, $\partial_{u_g}^2\al_a=(2\wp_{gg}+\lambda_{2g}+b_a)\al_a$, so that the third row of the Wronskian matrix may be replaced by $(b_a\al_a)_{a=1,2,3}$, and
$$
W(\al_1,\al_2,\al_3)=b_1\al_1W_{23}-b_2\al_2W_{13}+b_3\al_3W_{12}
=(b_1-b_3)\al_1W_{23}-(b_2-b_3)\al_2W_{13},
$$
where we used $\al_1W_{23}-\al_2W_{13}+\al_3W_{12}=0$.
Hence $W(\al_1,\al_2,\al_3)=(b_1-b_3)(b_3-b_2)(\al_1\al_{23}-\al_2\al_{13})$, and (\ref{eq:Baker123}) gives the assertion.
\end{proof}

\begin{remark}\label{rmk:BakerCrum}
{\rm{
Baker's determinant formulas for $\sigma^{n-1}\sigma_{12\cdots(2n)}$ and $\sigma^{n}\sigma_{12\cdots(2n+1)}$ \cite[Section IX]{Baker98} have the same structure, and it is strongly expected that
$$
\al_{a_1\cdots a_k}=c\,\frac{W(\al_{a_1},\ldots,\al_{a_k})}{\prod_{i<j}(b_{a_j}-b_{a_i})},\qquad c\in\CC^\times,
$$
holds for every $k$.
This means that the translation by $k$ half-periods $u\mapsto u+\omega_{B_{a_1}}+\cdots+\omega_{B_{a_k}}$ realizes the $k$-fold Darboux transformation given by Crum's Wronskian at the band edges $b_{a_1},\ldots,b_{a_k}$, i.e., $2\wp_{gg}(u)-2\partial_{u_g}^2\log W(\al_{a_1},\ldots,\al_{a_k})=2\wp_{gg}(u+\omega_{B_{a_1}}+\cdots+\omega_{B_{a_k}})$, and that Baker's determinant formulas are expansions of Crum's Wronskian.
}}
\end{remark}

\subsection{Algebraic and geometric properties of \texorpdfstring{$\al_{12}$}{al\textunderscore 12} functions}
\label{subs:geom_al12}

In order to express the $\al_{12}$ function precisely as in Proposition \ref{pr:al12=ww}, we introduce the relevant coverings of $X$.

\smallskip
\noindent\textbf{The covering on which $\al_{12}$ is single-valued is a double covering.}
The function field $\CC(X)(w_1,w_2)$ is a Galois extension of $\CC(X)$ with group $\ZZ/2\times\ZZ/2$ (a Klein four extension), and it has exactly three intermediate quadratic extensions,
$$
\CC(X)(w_1),\qquad \CC(X)(w_2),\qquad \CC(X)(\hw),\quad \hw:=w_1w_2,\quad \hw^2=(x-b_1)(x-b_2).
$$
Correspondingly, the fourfold covering $\hX\to X$ factors through three double coverings $\hX_1$, $\hX_2$ and $\hX_{12}$:
\begin{equation}
\xymatrix{ 
& \hX \ar[dl]\ar[d]\ar[dr] & \\
\hX_1\ar[dr]_{\varpi_{\hX,1}} & \hX_{12}\ar[d]^{\varpi_{\hX,12}} & \hX_2\ar[dl]^{\varpi_{\hX,2}}\\
 & X. &
}\label{2eq:tXtP}
\end{equation}
Since $\al_1\al_2=\prod_r w_{1,r}w_{2,r}=\prod_r \hw_r$ and $\Delta_{b_1,b_2}$ is a function on $S^gX$, the function $\al_{12}=\al_1\al_2\Delta_{b_1,b_2}$ already lives on $S^g\hX_{12}$; the fourfold covering $\hX$ is only needed when $\al_1$, $\al_2$ and $\al_{12}$ are handled at the same time.

\smallskip
We now describe $\hX_{12}$.
Let $\hz:=y/\hw$.
Then $\hz^2=y^2/((x-b_1)(x-b_2))=\prod_{j=3}^{2g+1}(x-b_j)$, and $\hX_{12}$ is the (normalization of the) curve in $\CC^3$ given by
\begin{equation}
\hw^2=(x-b_1)(x-b_2),\qquad \hz^2=\prod_{j=3}^{2g+1}(x-b_j),
\label{eq:hX12}
\end{equation}
with the projection $\varpi_{\hX,12}(x,\hw,\hz)=(x,\hw\hz)$, whose affine ring is
$$
R_{\hX,12}:=\CC[x,\hw,\hz]\big/\Bigl(\hw^2-(x-b_1)(x-b_2),\ \hz^2-\prod_{j=3}^{2g+1}(x-b_j)\Bigr).
$$
This ring is regular (the two equations have no common critical point because the $b_j$ are distinct), and it is the integral closure of $R_X[\hw]/(\hw^2-(x-b_1)(x-b_2))$, obtained by adjoining $\hz=y/\hw$.
The divisor of $(x-b_1)(x-b_2)$ on $X$ is $2(B_1+B_2-2\infty)$; hence $\hX_{12}\to X$ is the unramified double covering associated with the $2$-torsion point $v(B_1+B_2)$ of $J_X$, and by the Hurwitz formula its genus is $2g-1$.
Similarly the Klein four covering $\hX$ is unramified, of degree four and of genus $4g-3$; its affine ring is the integral closure of $R_X[w_1,w_2]/(w_1^2-x+b_1,\ w_2^2-x+b_2)$, which contains $z_1=y/w_1$, $z_2=y/w_2$ and $\hz=y/(w_1w_2)$, and which is a free $R_X$-module of rank four after inverting $(x-b_1)(x-b_2)$, with basis $1,w_1,w_2,w_1w_2$.

As in Section \ref{sec:Wcover}, we let
$$
\Gamma^{(12)}:=\varpi_{\hX,12*}\mathrm{H}_1(\hX_{12},\ZZ)\subset\Gamma_X,\qquad
\hJ_{X,12}:=\CC^g/\Gamma^{(12)},\qquad \hvarpi_{J,12}:\hJ_{X,12}\to J_X ,
$$
and we continue to use the Abelian integral $\tv:S^g\tX\to\CC^g$; $\tX$ is also an Abelian covering of $\hX_{12}$, $\kappa_{\hX,12}:\tX\to\hX_{12}$.
The lattice $\Gamma^{(12)}$ is the kernel of the monodromy character of $\hw$ on $\mathrm{H}_1(X,\ZZ)$.
Moving one point of the divisor around a cycle $\gamma$ changes $u$ by $\displaystyle{\oint_\gamma\nuI{}}$ and changes the sign of $\prod_r\hw_r$ by the monodromy of $\hw$ along $\gamma$; hence this character is the multiplier of $\al_{12}$ in Lemma \ref{lm:multiplier}, and $\Gamma^{(12)}=\{\ell\in\Gamma_X\,|\,\ell_1''\equiv0\ \mbox{modulo}\ 2\}$ has index two.
In contrast, the lattice of the Klein four covering $\hX$ is $\Gamma^{(1)}\cap\Gamma^{(2)}$, of index four.

\begin{proposition} \label{pr:al12=ww}

\begin{enumerate}
\item For a point $(\gamma_1, \ldots, \gamma_g) \in S^g \tX$ with $\kappa_{\hX,12}(\gamma_r)=(x_r,\hw_r,\hz_r)\in\hX_{12}$ $(r=1,\ldots,g)$ and $u=\tv(\gamma_1,\ldots,\gamma_g)$, we have
\begin{equation}
\al_{12}(u)=\Bigl(\prod_{s=1}^g \hw_s\Bigr)\sum_{r=1}^g\frac{\hz_r}{\hw_r\,F'(x_r)}
=\sum_{r=1}^g\frac{\hz_r\prod_{s\neq r}\hw_s}{F'(x_r)},
\qquad F'(x_r)=\prod_{s\ne r}(x_r-x_s),
\label{eq:al12=vt}
\end{equation}
which is a meromorphic function on $S^g\hX_{12}$, i.e., an element of the field of symmetric functions $(\cQ(R_{\hX,12})^{\otimes g})^{\fS_g}$.
Its lattice of periods is $\Gamma^{(12)}$, and $\al_{12}$ is a single-valued meromorphic function on the double covering $\hJ_{X,12}=\CC^g/\Gamma^{(12)}$ of $J_X$, but not on $J_X$.
\item On the Klein four covering, with $\kappa_{\hX}(\gamma_r)=(w_{1,r},z_{1,r},w_{2,r},z_{2,r})$, $\hw_r=w_{1,r}w_{2,r}$ and $\hz_r=z_{1,r}/w_{2,r}$, (\ref{eq:al12=vt}) reads
\begin{equation}
\al_{12}(u)=(w_{1,1}\cdots w_{1,g}) \cdot (w_{2,1}\cdots w_{2,g})
\sum_{r=1}^g\frac{z_{1,r}w_{1,r}}{w_{1,r}^2w_{2,r}^2 F'(w_{1,r}^2+b_1)},
\label{eq:al12=ww}
\end{equation}
and the three functions $\al_1=\prod_rw_{1,r}$, $\al_2=\prod_rw_{2,r}$ and $\al_{12}$ are single-valued simultaneously on $\CC^g/(\Gamma^{(1)}\cap\Gamma^{(2)})$.
\end{enumerate}

\end{proposition}

\begin{proof}
By (\ref{eq:al12norm}), Proposition \ref{ala=www} and $\hw_r=w_{1,r}w_{2,r}$, $\al_{12}=\al_1\al_2\Delta_{b_1,b_2}=(\prod_s\hw_s)\Delta_{b_1,b_2}$, and $\Delta_{b_1,b_2}=\sum_r y_r/((x_r-b_1)(x_r-b_2)F'(x_r))=\sum_r \hz_r/(\hw_rF'(x_r))$ since $y_r=\hw_r\hz_r$ and $(x_r-b_1)(x_r-b_2)=\hw_r^2$.
The statement on the lattice is Lemma \ref{lm:multiplier} together with the argument above.
(\ref{eq:al12=ww}) follows from $y=w_1z_1=w_2z_2$.
\end{proof}

\begin{remark} \label{rmk:al12=xxx}
{\rm{
Due to the double covering $\hX_{12}$ of $X$ and the above proposition, the ambiguity in the root expression $\al_{12}=\pm\sqrt{F(b_1)F(b_2)}\,\Delta_{b_1,b_2}$ disappears, and (\ref{eq:al12=vt}) together with the definition (\ref{eq:def_al12}) can be regarded as the Jacobi inversion formula for the covering $\hX_{12}$, in the same sense as Remark \ref{rmk:al1=xxx}.
Note that (\ref{eq:al12=vt}) involves only the coordinates $(x_r,\hw_r,\hz_r)$ of $\hX_{12}$; the square roots $w_1=\sqrt{x-b_1}$ and $w_2=\sqrt{x-b_2}$ appear individually only through $\al_1$ and $\al_2$.
}}
\end{remark}

\subsection{Differential relations of \texorpdfstring{$\al_{12}$}{al\textunderscore 12} functions I}

We present the differential properties of $\al_{12}$, which is partially reported in \cite[Section 4.2.2]{M25}:

As we showed in Theorem 4.49 in \cite{M25}, we have the following relations:
\begin{theorem}\label{4th:alab_phi}
Let
$$
\phi^\circ_{[12]}:=\log \al_{12},\quad
\fa^{\circ}_{12}
:=\frac{b_2-b_1}{\phi^{[2]\prime}-\phi^{[1]\prime}}.
$$

\begin{enumerate}

\item 
$
\displaystyle{
\al_{12}=\frac{1}{b_2-b_1}
\left(
\partial_{u_{g}} \phi^{[2]}
-\partial_{u_{g}} \phi^{[1]} \right)\al_1\al_2}
$,\ i.e., \
$
\displaystyle{
\al_{12}=\frac{\phi^{[2]\prime}-\phi^{[1]\prime}}{b_2-b_1}
\al_1\al_2}
$,

\item
$\displaystyle{
\partial_{u_{g}} \al_{12} = \al_1 \al_2}$,\  i.e.,\  
$\displaystyle{
\partial_{u_{g}} \al_{12} =
\left(\partial_{u_{g}}
\phi^\circ_{[12]}\right) \al_{12} 
}$, 
 i.e.,\  
$\displaystyle{
\partial_{u_{g}} \al_{12} =
\fa^{\circ}_{[12]}\al_{12} 
}$,

\item
$\displaystyle{
\partial_{u_{g}} \frac{\al_2}{\al_1} = (\phi^{[2]\prime}-\phi^{[1]\prime}) \frac{\al_2}{\al_1} }$,

\item
$\displaystyle{
\partial_{u_{g}} \frac{1}{\al_{12}} = -\fa^{\circ}_{12} \frac{1}{\al_{12}} }$,

\item 
$
\displaystyle{
\fa^{\circ}_{12}=
\partial_{u_{g}}
\phi^\circ_{[12]}=
\frac{\al_1\al_2}{\al_{12}}}$.

\end{enumerate}

\end{theorem}

\begin{proof}
Proposition \ref{3pr:Bol_rel} shows (1).
We consider (2):
$\displaystyle{
\partial_{u_{g}} \al_{12} =\frac{\phi^{[2]\prime\prime}-\phi^{[1]\prime\prime}}{b_2-b_1}
\al_1\al_2}$\break
$\displaystyle{
+ 
\frac{\phi^{[2]\prime}-\phi^{[1]\prime}}{b_2-b_1}
[\phi^{[2]\prime}+\phi^{[1]\prime}]\al_1\al_2
}$
$\displaystyle{
 =\frac{-\phi^{[2]\prime2}+\phi^{[1]\prime2}+b_2-b_1}{b_2-b_1}
\al_1\al_2+ 
\frac{\phi^{[2]\prime2}-\phi^{[1]\prime2}}{b_2-b_1}\al_1\al_2
}$ $=\al_1\al_2$.
Here we used Theorem \ref{4th:al_wp_phi}.
(3)-(5) are obvious.
\end{proof}

Relations of the type (2), in which a derivative of an al function is a product of two al functions, appear already in \cite[(8)]{Wei54}, where $\partial\,\mathrm{al}(u)_\alpha/\partial u_a$ is expressed by $\mathrm{al}(u)_{2a-1}\mathrm{al}(u)_{\alpha,2a-1}$; in genus one they reduce to $\sn'=\cn\,\dn$ and its companions.

Further, as a corollary of Theorem \ref{4th:alab_phi}, we also have the following relations as
 we showed in Proposition 4.50 in \cite{M25}:

\begin{proposition}\label{4pr:alab_phi}

Let 
$\displaystyle{\fD^{\circ}_{12}
:=\frac{(b_2-b_1)(\phi^{[2]\prime}+\phi^{[1]\prime})}{
 \phi^{[2]\prime}-\phi^{[1]\prime}}}$.

\begin{enumerate}

\item 
$\displaystyle{
\fD^{\circ}_{12}=
\fa^{\circ2}_{[12]}+
\partial_{u_{g}}\fa^\circ_{[12]}
}$,

\item
$\displaystyle{
\partial_{u_{g}}^2\fa^\circ_{[12]}
=2\fa^{\circ}_{[12]}(\fa^{\circ2}_{[12]}-
\fD^{\circ}_{12})+
\partial_{u_{g}}\fD^{\circ}_{[12]}
}$,

\item
$\displaystyle{
\fa^{\circ2}_{[12]}-
\fD^{\circ}_{12}
=\fa^{\circ2}_{[12]}\left(
\frac{(b_2-b_1)-
(\phi^{[2]\prime2}
-\phi^{[1]\prime2})}{b_2-b_1}\right)
}$,

\item 
$\displaystyle{
\partial_{u_{g}}\fD^{\circ}_{[12]}
=\fD^{\circ\prime}_{[12]}
=2\fa^{\circ2}_{[12]}
\frac{
\phi^{[2]\prime}\phi^{[1]\prime\prime}-
\phi^{[1]\prime}\phi^{[2]\prime\prime} }
{b_2-b_1}
}$.

\end{enumerate}

\end{proposition}

\subsection{Differential relations of \texorpdfstring{$\al_{12}$}{al\textunderscore 12} functions II}

Corresponding to Lemma \ref{lm:al12_omega}, we show three type identities as the differential identities of $\al_{12}$ functions by recalling the conventions, $\dot{h}:=\partial_{u_{g-1}}h$ and $h':=\partial_{u_{g}}h$, of Theorem \ref{4th:al_wp_phi}.
The following are our main theorems in this paper.

\begin{theorem}\label{th:5.6}
The following holds

\begin{enumerate}
\item
$\displaystyle{
\partial_{u_{g}}^2 \al_{12} =
\left(\partial_{u_{g}} \phi^{[2]}
+\partial_{u_{g}} \phi^{[1]} \right)\al_1 \al_2}$, i.e., 
$\displaystyle{
\partial_{u_{g}}^2 \al_{12} =
(b_2-b_1)\left[\frac{
\phi^{[2]\prime}
+\phi^{[1]\prime}}
{\phi^{[2]\prime}-\phi^{[1]\prime} }\right]
\al_{12}}$,\ or \break 
$\displaystyle{
\partial_{u_{g}}^2 \al_{12} =
\fD^{\circ}_{12}
\al_{12}}$.

\item 
$\displaystyle{
 \partial_{u_{g}}^2 \frac{\al_2}{\al_{1}}
-[2\phi^{[1]\prime}(\phi^{[1]\prime}-\phi^{[2]\prime})
+(b_2-b_1)]
\frac{\al_2}{\al_{1}}=0
}$.

\item
$\displaystyle{
 \partial_{u_{g}}^2 \frac{1}{\al_{12}}
+[\fD^{\circ}_{[12]}-2\fa^{\circ2}_{[12]}]\frac{1}{\al_{12}}=0
}$.

\end{enumerate}
\end{theorem}

\begin{proof}
(1) is obvious.
We consider (2):
$\displaystyle{
\partial_{u_{g}}^2 
\frac{\al_2}{\al_1}=
[\phi^{[2]\prime\prime}-\phi^{[1]\prime\prime}
+(\phi^{[2]\prime}-\phi^{[1]\prime})^2]
\frac{\al_2}{\al_1}
}$.
Theorem \ref{4th:al_wp_phi} (2), i.e., $\phi^{[2]\prime\prime}-\phi^{[1]\prime\prime}=-(\phi^{[2]\prime2}-\phi^{[1]\prime2})+(b_2-b_1)$, means that $\displaystyle{
\left[\frac{\al_2}{\al_1}\right]''
}$ is equal to 
$\displaystyle{
\bigl[[-\phi^{[2]\prime}-\phi^{[1]\prime}
+\phi^{[2]\prime}-\phi^{[1]\prime}]
(\phi^{[2]\prime}-\phi^{[1]\prime})
+(b_2-b_1)]
\frac{\al_2}{\al_1}
}$.

From Proposition \ref{4pr:alab_phi}, we have
$\displaystyle{
 \partial_{u_{g}}^2 \frac{1}{\al_{12}}=
-\partial_{u_{g}} \frac{\fa^\circ_{12}}{\al_{12}}
}$
$\displaystyle{
=-[\fD^\circ_{12}-\fa^{\circ2}_{12}]\frac{1}{\al_{12}}
}$
$\displaystyle{
-\fa^\circ_{12}(-\fa^\circ_{12} \frac{1}{\al_{12}})
}$.
Thus it is equal to
$\displaystyle{
[-\fD^\circ_{12}+2\fa^{\circ2}_{12}]\frac{1}{\al_{12}}
}$.
\end{proof}

Furthermore, we have the following differential identities as the second main theorem:
\begin{theorem}\label{th:5.7}
The following holds

\begin{enumerate}
\item 
$\displaystyle{
4\partial_{u_{g-1}}\al_{12}
-[2\lambda_{2g}+3(b_1+b_2)]\partial_{u_{g}}\al_{12}
+\left[6\phi^{[1]\prime}\phi^{[2]\prime}\right]
\partial_{u_{g}}\al_{12}
-\partial_{u_{g}}^3\al_{12}=0
}$,

\item
$\displaystyle{
4\partial_{u_{g-1}}\frac{\al_2}{\al_1}
-[2\lambda_{2g}+3(b_1+b_2)]\partial_{u_{g}}\frac{\al_2}{\al_1}
+6\phi^{[1]\prime2}
\partial_{u_{g}}\frac{\al_2}{\al_1}
-\partial_{u_g}^3\frac{\al_2}{\al_1}}$

${
-6(b_2-b_1)\phi^{[1]\prime}\frac{\al_2}{\al_1}=0,
}$

\item 
$\displaystyle{
4\partial_{u_{g-1}}\frac{1}{\al_{12}}
-[2\lambda_{2g}+3(b_1+b_2)]\partial_{u_{g}}\frac{1}{\al_{12}}}$

$\displaystyle{
+\left[6\fa^{\circ2}_{12}-6\fD^\circ_{12}+6\phi^{[1]\prime}\phi^{[2]\prime}\right]
\partial_{u_{g}}\frac{1}{\al_{12}}
-\partial_{u_{g}}^3\frac{1}{\al_{12}}=0
}$.

\end{enumerate}

\end{theorem}

\noindent
The identities (1)--(3) are equivalent to each other by Theorem \ref{th:KN} below.
They are identities between Abelian functions, and after the elimination of $\phi^{[a]\prime}$, $\fa^\circ_{12}$, $\fD^\circ_{12}$ they reduce to the known (degenerate) Krichever-Novikov equation (Theorem \ref{th:KN}).

\begin{proof}
(1): By using Theorem \ref{4th:al_wp_phi}, we have
\begin{gather*}
\begin{split}
 \partial_{u_{g-1}}\al_{12}&=
\left[[\dot\phi^{[2]}+\dot\phi^{[1]}]
\frac{\phi^{[2]\prime}-\phi^{[1]\prime}}{b_2-b_1}
+\frac{\dot\phi^{[2]\prime}-\dot\phi^{[1]\prime}}{b_2-b_1}
\right]\al_1\al_2\\
&=[\wp_{gg}-\phi^{[1]\prime}\phi^{[2]\prime}
+\lambda_{2g}+b_2+b_1]\al_1\al_2,
\end{split}
\end{gather*}
and
\begin{gather*}
\begin{split}
 \partial_{u_{g}}^3\al_{12}&=
\left[[\phi^{[1]\prime\prime}+\phi^{[2]\prime\prime}]
+[\phi^{[1]\prime}+\phi^{[2]\prime}]^2
\right]\al_1\al_2\\
&=[4\wp_{gg}+2\phi^{[1]\prime}\phi^{[2]\prime}
+2\lambda_{2g}+b_1+b_2]\al_1\al_2.
\end{split}
\end{gather*}
We combine them to obtain (1).

(2): We also use Theorem \ref{4th:al_wp_phi} to have
\begin{gather*}
\begin{split}
 \partial_{u_{g-1}}\frac{\al_2}{\al_1}&=
[\dot\phi^{[2]}-\dot\phi^{[1]}]\frac{\al_2}{\al_1}\\
&=-\wp_{gg}[\phi^{[2]\prime}-\phi^{[1]\prime}]\frac{\al_2}{\al_1}
+[b_2\phi^{[2]\prime}-b_1\phi^{[1]\prime}]\frac{\al_2}{\al_1},
\end{split}
\end{gather*}
and
\begin{gather*}
\begin{split}
 \partial_{u_{g}}^3\frac{\al_2}{\al_1}&=
\partial_{u_g}\Bigl(P\,\dfrac{\al_2}{\al_1}\Bigr),\qquad P:=-2\phi^{[1]\prime}d+(b_2-b_1),\quad d:=\phi^{[2]\prime}-\phi^{[1]\prime}\\
&=(P'+Pd)\frac{\al_2}{\al_1}\\
&=[6\phi^{[1]\prime2}-4\wp_{gg}-2\lambda_{2g}+b_2-3b_1]\left[\frac{\al_2}{\al_1}\right]^\prime
-2(b_2-b_1)\phi^{[1]\prime}
\left[\frac{\al_2}{\al_1}\right].
\end{split}
\end{gather*}
Here we used Theorem \ref{th:5.6}\,(2), $d'=-(\phi^{[1]\prime}+\phi^{[2]\prime})d+(b_2-b_1)$ and $\phi^{[1]\prime\prime}=2\wp_{gg}+\lambda_{2g}+b_1-\phi^{[1]\prime2}$.
In the first line, writing $4b_2\phi^{[2]\prime}-4b_1\phi^{[1]\prime}=4b_2d+4(b_2-b_1)\phi^{[1]\prime}$, we have $4\partial_{u_{g-1}}\dfrac{\al_2}{\al_1}=(4b_2-4\wp_{gg})\Bigl[\dfrac{\al_2}{\al_1}\Bigr]'+4(b_2-b_1)\phi^{[1]\prime}\dfrac{\al_2}{\al_1}$.
We combine them to obtain (2).

(3): Theorem \ref{4th:al_wp_phi} leads
\begin{gather*}
 \partial_{u_{g-1}}\frac{1}{\al_{12}}=-
\frac{\partial_{u_{g-1}}\al_{12}}{\al_{12}^2}
=[\wp_{gg}-\phi^{[1]\prime}\phi^{[2]\prime}
+\lambda_{2g}+b_2+b_1]
\left[\frac{1}{\al_{12}}\right]',
\end{gather*}
and
\begin{equation*}
\begin{split}
\partial_{u_{g}}^3 \frac{1}{\al_{12}}
&=\partial_{u_{g}}
\bigr[-\fD^\circ_{12}+2\fa^{\circ2}_{12}\bigr]\frac{1}{\al_{12}}\\
&=\bigr[-\fD^{\circ\prime}_{12}+4\fa^{\circ}_{12}\fa^{\circ\prime}_{12}\bigr]
\frac{1}{\al_{12}}
-\bigr[-\fD^\circ_{12}+2\fa^{\circ2}_{12}\bigr]\frac{\fa^{\circ}_{12}}{\al_{12}}\\
&=\bigl[4\wp_{gg}+2\lambda_{2g}+b_1+b_2+2\phi^{[1]\prime}\phi^{[2]\prime}
-6\fD^{\circ}_{12}+6\fa^{\circ2}_{12}\bigr]
\left[\frac{1}{\al_{12}}\right]'.
\end{split}
\end{equation*}
Here we used $\Bigl[\dfrac{1}{\al_{12}}\Bigr]'=-\dfrac{\fa^\circ_{12}}{\al_{12}}$, $\fa^{\circ\prime}_{12}=\fD^\circ_{12}-\fa^{\circ2}_{12}$ and
$\dfrac{\fD^{\circ\prime}_{12}}{\fa^\circ_{12}}=4\wp_{gg}+2\lambda_{2g}+b_1+b_2+2\phi^{[1]\prime}\phi^{[2]\prime}-\fD^\circ_{12}$, which follows from Proposition \ref{4pr:alab_phi}\,(4) and Theorem \ref{4th:al_wp_phi}\,(2).

We combine them to obtain (3).
\end{proof}

\subsection{Elimination of the auxiliary quantities: a degenerate Krichever-Novikov equation}

The identities of Theorem \ref{th:5.7} contain the auxiliary quantities $\phi^{[a]\prime}$, $\fa^\circ_{12}$ and $\fD^\circ_{12}$.
Eliminating them, we obtain the Krichever-Novikov equation (\ref{eq:KN_intro}) in a degenerate form.

\begin{theorem}\label{th:KN}
Let $s:=u_g$, $\partial_t:=4\partial_{u_{g-1}}-\bigl(2\lambda_{2g}+3(b_1+b_2)\bigr)\partial_{u_g}$ and $\beta:=b_2-b_1$.
Each of the three functions
$$
q=\al_{12},\qquad q=\frac{\al_2}{\al_1},\qquad q=\frac{1}{\al_{12}}
$$
satisfies the degenerate Krichever-Novikov equation
\begin{equation}
q_t=q_{sss}-\frac32\,\frac{q_{ss}^2-\beta^2q^2}{q_s},
\label{eq:KNal}
\end{equation}
i.e., (\ref{eq:KN_intro}) with $h(q)=(b_1-b_2)^2q^2$, whose roots (as a binary quartic) are $0$ and $\infty$, both double.
Conversely, Theorem \ref{th:5.7}\,(1), (2) and (3) are equivalent to (\ref{eq:KNal}) for $q=\al_{12}$, $\al_2/\al_1$ and $1/\al_{12}$, respectively, together with Theorems \ref{4th:alab_phi} and \ref{th:5.6}.
\end{theorem}

\begin{proof}
For $q=\al_{12}$: by Theorem \ref{4th:alab_phi}\,(2) and Theorem \ref{th:5.6}\,(1), $q_s=\al_1\al_2$ and $q_{ss}/q_s=\phi^{[1]\prime}+\phi^{[2]\prime}$, and by Theorem \ref{4th:alab_phi}\,(1), (5), $\phi^{[2]\prime}-\phi^{[1]\prime}=\beta q/q_s$.
Hence
$$
\phi^{[1]\prime}\phi^{[2]\prime}=\frac14\Bigl[(\phi^{[1]\prime}+\phi^{[2]\prime})^2-(\phi^{[2]\prime}-\phi^{[1]\prime})^2\Bigr]=\frac{q_{ss}^2-\beta^2q^2}{4q_s^2},
$$
and Theorem \ref{th:5.7}\,(1) becomes (\ref{eq:KNal}).

For $q=\al_2/\al_1$: $q_s/q=\phi^{[2]\prime}-\phi^{[1]\prime}$, and Theorem \ref{th:5.6}\,(2) gives $q_{ss}=\bigl(-2\phi^{[1]\prime}(\phi^{[2]\prime}-\phi^{[1]\prime})+\beta\bigr)q$, i.e., $\phi^{[1]\prime}=(\beta q-q_{ss})/(2q_s)$.
Substituting this into Theorem \ref{th:5.7}\,(2), the terms $\pm3\beta qq_{ss}/q_s$ cancel and we obtain (\ref{eq:KNal}).

For $q=1/\al_{12}$: since $\{1/q;s\}=\{q;s\}$ and $h(1/q)q^4=h(q)$ for $h(q)=\beta^2q^2$, (\ref{eq:KNal}) is invariant under $q\mapsto1/q$; it is also directly obtained from Theorem \ref{th:5.7}\,(3) by using $q_s=-\fa^\circ_{12}q$ and $q_{ss}=(2\fa^{\circ2}_{12}-\fD^\circ_{12})q$ (Theorem \ref{th:5.6}\,(3)) and $\phi^{[1]\prime}\phi^{[2]\prime}=(\fD^{\circ2}_{12}-\beta^2)/(4\fa^{\circ2}_{12})$.
\end{proof}

\begin{remark}\label{rmk:KN}
{\rm{
(1) Equation (\ref{eq:KNal}) is invariant under $q\mapsto cq$ $(c\in\CC^\times)$ and $q\mapsto 1/q$, i.e., under the M\"obius transformations preserving the root set $\{0,\infty\}$ of $h$.
This is consistent with the half-period relations: $\al_{12}(u+\omega_{B_1})=c_{12,1}\,\al_2(u)/\al_1(u)$ (Lemma \ref{lm:al12_omega}\,(2)) and $\al_{12}(u+\omega_{12})=c/\al_{12}(u)$ (Lemma \ref{lm:al12_omega}\,(3)), and the translation $u\mapsto u+\omega$ commutes with $\partial_t$ and $\partial_s$.
Thus the three identities of Theorem \ref{th:5.7} are one identity seen at three points of the orbit of the half-periods.

(2) The functions $q=\al_2/\al_1$ and $\widetilde q=\al_{12}$ are related by
$$
q_s\,\widetilde q_s=\beta\,q\,\widetilde q .
$$
Indeed, by Theorem \ref{4th:alab_phi}\,(1), (2), $q_s=(\phi^{[2]\prime}-\phi^{[1]\prime})q=\beta\,\al_{12}/\al_1^2$ and $\widetilde q_s=\al_1\al_2$.
This is Adler's B\"acklund transformation $q_s\widetilde q_s=p(q,\widetilde q)$ \cite{Adler98} with $p=\beta q\widetilde q$, for which both discriminants $p_{\widetilde q}^2-2pp_{\widetilde q\widetilde q}=\beta^2q^2$ and $p_q^2-2pp_{qq}=\beta^2\widetilde q^2$ coincide with $h$ in (\ref{eq:KNal}).
Equivalently, since $\al_{12}$ is the Wronskian of $\al_1$ and $\al_2$ (Remark \ref{rmk:Crum}), the relation says that the ratio and the Wronskian of two band-edge eigenfunctions solve the same KN equation.

(3) The degenerate KN equation is mapped to the KdV equation by the substitution of Svinolupov, Sokolov and Yamilov \cite[(5)]{SSY83}, which for (\ref{eq:KNal}) reads $\mathfrak{v}=-3q_{sss}/q_s+\frac32(q_{ss}/q_s)^2-\frac32\beta^2(q/q_s)^2$.
For $q=\al_2/\al_1$, Theorem \ref{4th:al_wp_phi}\,(3) and Remark \ref{rmk:Crum} give
$$
\mathfrak{v}=-12\wp_{gg}(u+\omega_{B_1})-3\bigl(2\lambda_{2g}+b_1+b_2\bigr),
$$
since, for two eigenfunctions $L\psi_a=E_a'\psi_a$ of a Schr\"odinger operator $L=\partial_s^2-U$, the substitution applied to $\psi_2/\psi_1$ gives $-6\bigl(U-2(\log\psi_1)_{ss}\bigr)-3(E_1'+E_2')$, which is a direct computation.
Hence the solutions of Theorem \ref{th:KN} correspond, under the substitution of \cite{SSY83}, to the translates of the finite-gap potential $2\wp_{gg}$, in the same way as $\phi^{[a]\prime}$ corresponds to $2\wp_{gg}$ by the Miura transformation (Theorem \ref{4th:al_wp_phi}\,(2)).
}}
\end{remark}

\section{Discussion}

Following Section 2, we defined the $\al_a$ and $\al_{12}$ functions in Definition \ref{def:als} and gave their basic properties in Section 3.
Lemmas \ref{lm:al_a}, \ref{lm:al12_omega} and \ref{lm:Delta_ab} are reported here for the first time.
In Section 4 we reviewed the algebraic and geometric properties of the $\al_a$ functions and their differential identities \cite{M25}, and in Section 5 we investigated $\al_{12}$.
Theorems \ref{th:5.6}, \ref{th:5.7} and \ref{th:KN} are the main results of this paper.
The geometric content of Section 5 is that $\al_{12}$ is a meromorphic function on $S^g\hX_{12}$ for the unramified double covering $\hX_{12}$: $\hw^2=(x-b_1)(x-b_2)$ of $X$, and is single-valued on the double covering $\CC^g/\Gamma^{(12)}$ of $J_X$ (Proposition \ref{pr:al12=ww}).
The analytic content is that $\al_{12}$ is the Wronskian of the band-edge eigenfunctions $\al_1$ and $\al_2$ (Remark \ref{rmk:Crum}), and that $\al_{12}$, $\al_2/\al_1$ and $1/\al_{12}$ satisfy the degenerate Krichever-Novikov equation (\ref{eq:KNal}).
These identities are identities between Abelian functions; we do not claim that the equation is new.
Their meaning is that the KN equation plays for the $\al_{ab}$ functions the role that the MKdV and sine-Gordon equations play for the $\al_a$ functions.

For comparison, we recall the sine-Gordon equation for the $\al_a$ functions (cf. \cite{Mat02a, M25}, \cite[3.241]{Mum84}).
In our normalization it reads, with $D_{b_1}:=\sum_{i=1}^g b_1^{\,i-1}\partial_{u_i}$,
\begin{gather}
D_{b_1}\partial_{u_g} \log \al_1^{2}
   = 2(-1)^{g+1}\Bigl(\al_1^2 - \frac{f'(b_1)}{\al_1^2}\Bigr)
   = 2(-1)^{g+1}\bigl(\al_1(u)^2 - \al_1(u+\omega_{B_1})^2\bigr), \label{4eq:SG1-4a}
\end{gather}
where the last equality is Lemma \ref{lm:al_a}\,(7).
Here the half-period translation $u\mapsto u+\omega_{B_1}$ relates the two terms on the right-hand side, just as the half-period translations relate the three solutions $\al_{12}$, $\al_2/\al_1$ and $1/\al_{12}$ of (\ref{eq:KNal}).

\bigskip

Several questions remain.
First, we have treated $\al_{12}$; by Remark \ref{rmk:3.6} the results extend to $\al_{ab}$ for any pair of finite branch points.
The functions $\al_{a_1\cdots a_k}$ associated with Baker's shifts $\sigma(u+\omega_{B_{a_1}}+\cdots+\omega_{B_{a_k}})$ are Crum's Wronskians of $k$ band-edge eigenfunctions for $k=3$ (Proposition \ref{pr:Baker_Crum}) and are strongly expected to be so in general (Remark \ref{rmk:BakerCrum}); it is natural to ask which integrable equations they satisfy, as (\ref{eq:KNal}) for $k=2$.
Second, the KN equation (\ref{eq:KNal}) is degenerate because its solutions are built from the eigenfunctions at two band edges of the same Schr\"odinger operator.
In the non-degenerate case the algebro-geometric solutions are constructed on a double covering of an elliptic curve by a Prym variety \cite{Novikov99}; since $\al_{12}$ is defined on the double covering $\hX_{12}$ of $X$, it would be interesting to understand both cases from a common point of view.
Third, the relation with the theta-function solutions of the Schwarzian KdV equation obtained from degenerations of Fay's identity \cite{KP23} suggests that (\ref{eq:KNal}) may also be derived from a suitable degeneration of Fay's trisecant identity with two branch points.
Finally, as the Jacobi elliptic functions are used in many applications, the real reductions of the $\al_{ab}$ functions, which require fixing the parameters $b_i$ of the curve as for the $\al_a$ functions in \cite{Ma24b, Ma26}, are left for future study.

The al functions can also be defined for non-hyperelliptic curves \cite{MP15, M25}, where they play an essential role in the description of degenerating families of curves \cite{FKMPA, M25}.
As this paper reveals the relations of the $\al_{12}$ functions with the double coverings, the Darboux transformation and the Krichever-Novikov equation, we believe that the $\al_a$ and $\al_{ab}$ functions themselves provide useful tools for understanding hyperelliptic curves and the integrable equations associated with them.

\appendix

\section{The genus one case}\label{sec:g1}

In this appendix we show how the results of Section 5 look in genus one, where they reduce to classical properties of the Jacobi elliptic functions.
Let $X$ be the elliptic curve $y^2=(x-b_1)(x-b_2)(x-b_3)=x^3+\lambda_2x^2+\lambda_1x+\lambda_0$, $\displaystyle{u=\int^{(x,y)}_\infty \frac{dx}{2y}}$, so that $x=\wp_{11}(u)$ and $\dfrac{d}{du}=2y\dfrac{d}{dx}$.
All the entities of the main text apply with $g=1$; in particular $\al_a(u)=\sqrt{x-b_a}$ and, by (\ref{eq:al12norm}) and (\ref{eq:Delta}),
$$
\al_{12}=\al_1\al_2\,\frac{y}{(x-b_1)(x-b_2)}=\frac{y}{\al_1\al_2}=\pm\al_3 ,
$$
which reflects the congruence $\omega_{B_1}+\omega_{B_2}\equiv\omega_{B_3}$ modulo the lattice.
With $z=\sqrt{b_1-b_3}\,u$, $k^2=(b_2-b_3)/(b_1-b_3)$ and $k'^2=1-k^2=(b_1-b_2)/(b_1-b_3)$, (\ref{eq:Jsn}) gives
$$
\frac{1}{\al_{12}}=\pm\frac{\sn(z)}{\sqrt{b_1-b_3}},\qquad
\al_{12}=\pm\sqrt{b_1-b_3}\,\mathrm{ns}(z),\qquad
\frac{\al_2}{\al_1}=\mathrm{dc}(z).
$$
The relations of Theorem \ref{4th:alab_phi} become $\partial_u\al_3=\pm\al_1\al_2$, which is equivalent to $\sn'=\cn\,\dn$.

In genus one there is no variable $u_{g-1}$, and the KN equation of Theorem \ref{th:KN} reduces to the ordinary differential equation
\begin{equation}
q'''-\frac32\,\frac{(q'')^2-(b_1-b_2)^2q^2}{q'}+\bigl(2\lambda_2+3(b_1+b_2)\bigr)q'=0,\qquad q':=\frac{dq}{du},
\label{eq:g1KN}
\end{equation}
where $2\lambda_2+3(b_1+b_2)=b_1+b_2-2b_3$.
In the variable $z$ it reads
\begin{equation}
q_{zzz}-\frac32\,\frac{q_{zz}^2-k'^4q^2}{q_z}+(1+k^2)\,q_z=0,
\label{eq:g1KNz}
\end{equation}
and it is satisfied by $q=\sn(z)$, $q=\mathrm{ns}(z)$ and $q=\mathrm{dc}(z)$, as is checked directly from $\sn'=\cn\,\dn$, $\cn'=-\sn\,\dn$, $\dn'=-k^2\sn\,\cn$.
Equation (\ref{eq:g1KNz}) is invariant under $q\mapsto cq$ and $q\mapsto1/q$, which exchanges $\sn$ and $\mathrm{ns}$, and the half-period translation $z\mapsto z+K$ maps $\mathrm{ns}$ to $\mathrm{dc}$, in accordance with Remark \ref{rmk:KN}\,(1).
Finally, the substitution of Remark \ref{rmk:KN}\,(3) maps $q=\al_2/\al_1$ to $\mathfrak{v}=-12\,x(u+\omega_{B_1})-3(2\lambda_2+b_1+b_2)$ with $x(u+\omega_{B_1})=b_1+(b_1-b_2)(b_1-b_3)/(x-b_1)$, a translate of the Lam\'e potential $2\wp_{11}$.
Thus the hyperelliptic $\al_{ab}$ functions and Theorem \ref{th:KN} generalize these elementary properties of the Jacobi elliptic functions to every genus, in the same way as the $\al_a$ functions generalize $\sn$, $\cn$ and $\dn$.

\section*{Acknowledgments}

The author would like to thank the late Professor Emma Previato for her valuable and helpful suggestions regarding this study in its early stages and Professor Jiryo Komeda for his valuable insights on covering curves.
The author has been supported by the Grant-in-Aid for Scientific Research (C) of the Japan Society for the Promotion of Science Grant, No.~21K03289 and No.~26K06824.

\bigskip

\noindent
Shigeki Matsutani\\
Electrical Engineering and Computer Science,\\
Graduate School of Natural Science \& Technology, \\
Kanazawa~University,\\
Kakuma Kanazawa, 920-1192, Japan\\
\texttt{ORCID:0000-0002-2981-0096}\\
\texttt{s-matsutani@se.kanazawa-u.ac.jp}

\end{document}